\documentclass[11pt,a4paper]{article}
\usepackage[top=60pt,bottom=60pt,left=82pt,right=82pt]{geometry}

\title{Distributed Dispatching in the Parallel
Server Model}

\author{ 
    Guy Goren\\
	\texttt{\small Technion}
	\and 
	Shay Vargaftik\\
	\texttt{\small VMware Research}
	\and 
	Yoram Moses\\
	\texttt{\small Technion}
}

\usepackage[utf8]{inputenc}
\usepackage{xcolor}
\usepackage{algorithm}
\usepackage[noend]{algpseudocode}
\usepackage{amssymb}
\usepackage{amsmath}
\usepackage{nccmath}
\usepackage{amsthm}
\usepackage{thmtools}
\usepackage{thm-restate}
\usepackage{hyperref}
\usepackage[capitalise]{cleveref}
\usepackage[normalem]{ulem}
\usepackage{mathtools}
\usepackage{empheq}
\usepackage{wrapfig}
\usepackage{fontawesome}
\usepackage{subcaption}  
\usepackage{graphicx}

\newcommand{\brec}[2]{\bar{\receives}_{#1}^{(#2)}}
\newcommand{\veca}{\vec{a}}
\newcommand{\LambdaOne}{\lambda^{(1)}}
\newcommand{\LambdaTwo}{\lambda^{(2)}}
\newcommand{\MuOne}{\mu^{(1)}}
\newcommand{\MuTwo}{\mu^{(2)}}
\newcommand{\setD}{{\cal D}}
\newcommand{\setS}{{\cal S}}
\newcommand{\receives}{g}
\newcommand{\sent}[2]{g^{(#1)}_{#2}}
\newcommand{\serves}{s}
\newcommand{\underWL}{\receives^*_{\scalebox{.4}{\faPlus}}}
\newcommand\set[1]{\{#1\}}

\newcommand\T[1]{\noindent\textbf{#1}}
\newcommand{\Binomial}[2]{\text{Bin}\left(#1,#2\right)}
\newcommand{\WFiE}{WFiE}

\newcommand{\brac}[1]{\left\{ #1 \right\}}

\newcommand{\naturals}{\mathbb{N}}
\newcommand{\setM}{\setD}
\newcommand{\setU}{\mathcal{U}}
\newcommand{\wl}{\textsc{WL}}
\newcommand{\wlf}[2]{\textsc{WaterLevel}(#1,#2)}
\newcommand{\setN}{\setS}

\newcommand{\setL}{{\mathcal{L}}} 

\newcommand{\bp}[1]{\Big(#1\Big)}

\newcommand{\TWLs}{\textsc{TWF}}

\newcommand{\sTWF}{s\textsc{TWF}}
\newcommand{\uTWF}{u\textsc{TWF}}

\providecommand{\ie}{\emph{i.e.,} }
\providecommand{\eg}{\emph{e.g.,} }

\DeclareMathOperator{\E}{\mathbb{E}}

\newtheorem{lemma}{Lemma}
\newtheorem{definition}{Definition}

\begin{document}
\date{}
\maketitle


\begin{abstract}
With the rapid increase in the size and volume of cloud services and data centers, architectures with multiple job dispatchers are quickly becoming the norm. Load balancing is a key element of such systems. Nevertheless, current solutions to load balancing in such systems admit a paradoxical behavior in which more accurate information regarding server queue lengths degrades performance due to herding and detrimental incast effects. Indeed, both in theory and in practice, there is a common doubt regarding the value of information in the context of multi-dispatcher load balancing. As a result, both researchers and system designers resort to more straightforward solutions, such as the power-of-two-choices to avoid worst-case scenarios, potentially sacrificing overall resource utilization and system performance. 
A principal focus of our investigation concerns the value of information about queue lengths in the multi-dispatcher setting. We argue that, at its core, load balancing with multiple dispatchers is a distributed computing task. In that light, we propose a new job dispatching approach, called \emph{Tidal Water Filling}, which addresses the distributed nature of the system.
Specifically, by incorporating the existence of other dispatchers into the decision-making process, our protocols outperform previous solutions in many scenarios. 
In particular, when the dispatchers have complete and accurate information regarding the server queues, our policies significantly outperform all existing solutions.
\end{abstract}


\section{Introduction}
\label{sec:Intro}
Large software systems that govern the operations of data centers and of cloud-based services are important components of modern-day computing infrastructure. Such systems process a large volume of jobs that often arrive at distinct locations and are serviced by a multitude of servers.  How jobs are assigned to servers, and in particular load balancing the servers' job queues, plays a central role in determining the effectiveness of the system's overall performance.

The traditional approach to load balancing in this setting, known as the \textit{supermarket model}  \cite{weber1978optimal,winston1977optimality,eryilmaz2012asymptotically}, employs a single centralized dispatcher to which all requests are forwarded, and from which they are assigned to the servers. In the last decade, with the increasing size of cloud services and applications, using a single dispatcher has become a problematic bottleneck. System designers have consequently shifted to architectures that employ multiple dispatchers \cite{gandhi2014duet,eisenbud2016maglev,prekas2017zygos}. 
The load balancing problem for such multi-dispatcher systems is a very natural and urgent distributed systems problem. However, to the best of our knowledge, it has received only limited attention in the literature (see Section~\ref{sec:related_work}). 
The current paper considers the problem from a distributed computing perspective.
We consider a setting  with~$M$ dispatchers and $N$ servers.   Jobs arrive at the dispatchers according to a stochastic process, and the servers complete jobs at a stochastic rate. 
Each dispatcher observes the jobs that it receives, but not those received by the other dispatchers. Dispatchers interact with the servers when they send them jobs, and can also communicate with them in order to obtain information regarding the lengths of the servers' queues. 

Load balancers are typically evaluated in terms of their behavior under heavy load, which occurs when the rate of job arrivals approaches the rate at which servers are able to process the jobs. 
A central factor in determining the quality of a load balancing protocol is the response time that it offers the clients. A natural measure is thus the expected response time (or {\it latency}) for jobs submitted to the system. In some cases, however, a client's task is broken up into multiple jobs, and the task is successfully served only when the last of these jobs is processed. 
For such instances, it is also important to meet a desired tail latency for the 95th, 99th, or even the 99.9th percentile of the distribution \cite{dean2013tail,nishtala2017hipster}.
Indeed, according to publications by Google and Amazon, even a small sub-second addition to response time in dynamic content websites has led to a persistent loss of users and revenue \cite{lu2011join,schurman2009user}.

A principal focus of our investigation concerns the value of information about queue lengths in the multi-dispatcher setting. As discussed below, several well-known load balancing policies behave poorly when the dispatchers' information is highly correlated. They suffer from so-called ``herd behavior,'' in which, when dispatchers identify the same servers as having short queues, they all send their jobs to these servers. As a consequence, the servers become overloaded, resulting in poor performance. This is especially acute when dispatchers share considerable information about the queue lengths. Indeed, in describing their solutions, recent papers have made statements such as
``{\em Inaccurate information can lead to better performance}''~\cite{vargaftik2020lsq}, or  ``{\em Inaccurate information can improve performance}''
\cite{zhou2020asymptotically}.  
Interestingly, the value of accurate information is doubted not only by theoreticians.
In fact, open source load balancer deployments such as HAProxy \cite{haproxy_po2} and NGYNX \cite{ngynx_po2} as well as cloud service companies such as Fastly \cite{youtube_lec} and Netflix \cite{netflixEdge} report making use of limited queue-size information in order to avoid detrimental herd behavior effects. 
This all appears to be rather perplexing and counter-intuitive. 
Is it indeed the case that there is such thing as too much information in the load-balancing context? Na\"{i}vely, at least, we would expect information to be valuable in a resource allocation context such as this. 

In this paper, we consider why popular policies suffer from herd behavior and suggest new policies that avoid it. We argue that in  policies that exhibit herd behavior, a dispatcher can be thought of as optimizing its behavior w.r.t.\ its information in a manner that is oblivious to the presence of other dispatchers. We suggest that this can be avoided by an analysis that explicitly accounts for the fact that other dispatchers are present. 
To focus on these issues, we start by considering an idealized setting in which all dispatchers have complete information regarding the queue sizes. Each dispatcher has access to the job requests that it receives, but does not know how many jobs each of the other dispatchers receives. (The distributions of arrival rates at dispatchers are unknown but assumed to be the same, as are those of server processing rates.) In this setting, we propose a new load-balancing policy called  {\em Tidal Water Filling} (TWF), and prove that it is in a precise sense optimal for the case of complete information. 
Indeed, we show that it improves on all known policies.
Finally, we demonstrate that its performance improves as the amount of information available to the dispatchers increases.
{\bf These results establish  that, when used appropriately, information can be gainfully used for load balancing in the multi-dispatcher setting}. These are the main contributions of the paper. 

Since the performance of Tidal Water Filling improves with available information regarding queue sizes, we turn to the question of how limited communication can be exploited to boost this information, thereby further improving the performance of TWF. To this end, we design a variant of the \TWLs{} protocol in which both servers and dispatchers keep track of queue size information. Moreover, whenever a dispatcher interacts with a server, the two share the information they have. The information is consistently updated by using standard distributed systems techniques such as timestamping the size data. 
We show using simulations that this approach allows a significant increase in the amount of relevant data used by dispatchers even when communication is limited, 
and, in turn, markedly improves the load balancing performance.

\subsection{Related Work}\label{sec:related_work}

In the standard supermarket (\ie parallel server) model, it is well known that if the (single) dispatcher has complete information regarding the server queues, the protocol that routes each job to the server with the shortest queue (called {\em Join the Shortest Queue} and denoted by JSQ) offers strong performance and strong theoretical guarantees \cite{weber1978optimal,winston1977optimality,eryilmaz2012asymptotically}. JSQ has also motivated the design of reduced-state load balancing techniques for resource-constrained scenarios in which the dispatcher is exposed to only partial information about the server queue lengths. For example, in {\em Power-of-$d$-choices}, denoted by JSQ($d$) \cite{luczak2006maximum,vvedenskaya1996queueing,mitzenmacher2001power}, when a job arrives, a dispatcher randomly probes $d$ servers and assigns the job to a server with the shortest queue among them.
A related strategy is called {\it  Power of memory}, denoted by  JSQ($d,m$) \cite{shah2002use,mitzenmacher2002load}. In JSQ($d,m$), the dispatcher samples the $m$ shortest queues to whom it sent jobs in the latest round,  in addition to $d \ge m \ge 1$ new  randomly chosen servers. The job is then routed to the shortest among these $d + m$ queues.

In the last decade, with the increasing size of cloud services and applications, the need to scale horizontally drove system designers to introduce multiple dispatchers into their design as a single dispatcher could no longer utilize hundreds and thousands of servers \cite{gandhi2014duet,eisenbud2016maglev,prekas2017zygos}. In such multi-dispatcher systems, traditional solutions such as JSQ suffer from detrimental herd behaviour and therefore systems operators abandon the use of readily available information and turn to reduced-state approaches such as JSQ($d$) potentially sacrificing overall system performance and reduced resource utilization in order to prevent worst-case scenarios \cite{netflixEdge,haproxy_po2,ngynx_po2,youtube_lec,lu2011join,schurman2009user}. 

In the search for a better alternative for the multi-dispatcher load balancing scenario, Join-the-idle-queue (JIQ) \cite{stolyar2015pull,lu2011join,mitzenmacher2016analyzing,stolyar2017pull,wang2018distributed} was recently proposed.
In JIQ, dispatchers are notified only by idled servers.
In turn, a dispatcher sends jobs to an idle server when it is aware of one, or to a randomly selected server otherwise. JIQ was shown to significantly improve performance at low and moderate loads over JSQ($d$)  due to its immediate prevention of server starvation~\cite{lu2011join}. However, at higher loads, its performance resembles random routing due to the absence of idle servers and its performance deteriorates quickly \cite{zhou2019heavy}. 
The Persistent-Idle load-distribution policies recently addressed this drawback of JIQ \cite{atar2020persistent,atar2021persistent}. However, so far, these policies have been proposed and analyzed only for a single dispatcher system. 

The most recent advances on load balancing for multi-dispatcher systems appeared in~\cite{vargaftik2020lsq,zhou2020asymptotically}. 
In the Local Shortest Queue policy (LSQ), proposed in~\cite{vargaftik2020lsq}, each dispatcher keeps a local state with possibly outdated queue size information, which is infrequently updated.
A dispatcher then sends jobs to a shortest queue by its local estimation.
This can be viewed as a generalization of similar ideas suggested for single-dispatcher systems \cite{anselmi2020power,van2019hyper}.
LSQ was followed by LED~\cite{zhou2020asymptotically}, which extended the theoretical performance guarantees to a wider family of \textit{tilted} dispatching policies. 
These local-state-driven policies were shown to outperform previous policies considered for the multi-dispatcher model such as JSQ, JSQ($d$) and JIQ. Intuitively, this is because they use considerably less information than JSQ, which reduces herding, and, on the other hand, maintain local states at the dispatchers allowing for a long term memory and preventing many events in which no single good server is discovered.

However, the aforementioned policies admit a paradoxical behavior in which accurate information, or even partial but correlated information among the dispatchers, degrades their performance.
This is because, like their complete state-information JSQ counterpart, having shared information about good servers leads them to herd-behavior and detrimental incast effects that increase tail latency.
That is, in the multiple-dispatchers case, when updated information is available to different dispatchers, they all send at once all incoming jobs to the servers with the currently-shorter queues, overwhelming them with the accumulated traffic. 
This phenomenon has already been pointed out in \cite{mitzenmacher2000useful}, which suggested the importance of using randomness to break the symmetry.

Another line of work concerning load balancing in distributed systems is based on the balls-into-bins model \cite{adler1998parallel}. 
Recent approaches commonly apply regret minimization (\eg \cite{kleinberg2011load}) and adaptive techniques (\eg \cite{lenzen2011tight}), producing more precise theoretical bounds.
However, their model assumptions are not aligned with our model (\eg we consider stochastic arrivals at each dispatcher and stochastic departures at each server).
Consequently, their analysis does not directly apply in our model and vice versa.


\section{Model}
\label{sec:Model}

We consider a  system with a set $\setD$ of~$M$ dispatchers and a set $\setS$ of~$N$ servers. Dispatchers can communicate with servers over communication channels or via shared memory, but there is no direct communication among dispatchers.
The network is the complete undirected bipartite graph with edge set $E=\setD\times\setS$, and both jobs and standard messages can be sent over the edges of~$E$.%
\footnote{The high rate of incoming jobs at the dispatchers makes interaction among them undesirable. As a result, the assumption that dispatchers do not interact is standard practice \cite{mitzenmacher2000useful,zhou2020asymptotically,vargaftik2020lsq,wang2018distributed,stolyar2017pull,mitzenmacher2016analyzing,stolyar2015pull,lu2011join,haproxy_po2,ngynx_po2}. }
The system proceeds over discrete synchronous rounds. Time starts at time~0, and round~$t+1$ occurs between time~$t$ and time~$t+1$. 
Each round consists of four phases: First, every dispatcher receives an external input with a set of job requests. Second, every dispatcher sends each of the jobs it received to a server for processing, and every job received by a server is added to its job queue. 
Third, each of the servers completes processing a set of (zero or more) jobs from the head of its queue and reports the results to the appropriate clients. (This is where jobs depart from the system.) Finally, in a potential fourth phase of the round, dispatchers and servers may communicate information about the status of the server queues. More formally, the four phases are:

\begin{enumerate}
    \itemsep1.4em
    \item \T{Arrivals:} Some number, $a^{(m)}(t)$, of exogenous jobs arrive at dispatcher~$m$ at the beginning of round~$t$. (We denote $a(t) = \sum_{m \in \setD} a^{(m)}(t)$.) Job arrivals are governed by stochastic processes. For simplicity, we assume that  $a^{(m)}(t)$ are {\em i.i.d.} random variables governed by the same distribution which, according to standard practice, \emph{is not assumed to be known}. 
    
    \item \T{Dispatching.} In every round, each dispatcher forwards the jobs it received to the servers for service. We consider two variants of dispatching that address two distinct affinity constraints.
    In one, termed {\em splittable dispatching}, the dispatcher assigns each of the jobs to a server of its choice.
    This handles a setting with no affinity constraints. That is, when the jobs arriving at a dispatcher~$m$ can be processed independently.
    We separately consider {\em unsplittable dispatching}, in which all jobs that arrive at~$m$ in a given round must be forwarded to the same server. 
    This handles a setting with strict affinity constraints.
    Where, \eg the jobs arriving at~$m$ in a given round share data or resources.

    \item \T{Departures.} Each server maintains a FIFO queue that keeps track of its pending jobs. We denote by~$Q_n(t)$ the  length of server~$n$'s queue at the beginning of round $t$ (before any job arrivals and departures) and denote $Q(t) \triangleq  \langle Q_1(t),\dots,Q_N(t)\rangle$. 
    Moreover, we denote by $\sent{m}{n}(t)$ the number of jobs that server~$n$ receives from dispatcher~$m$ in round~$t$. 
    Moreover, $\receives_n(t)=\sum_{m\in\setD}\sent{m}{n}(t)$ denotes the total number of jobs sent to server~$n$ in round~$t$. 
    We denote server~$n$'s job completion rate (i.e., the number of jobs that it is able to complete)  in round~$t$ by $\serves_n(t)$. 
    We thus have that  $Q_n(t+1)= \max\{0,Q_n(t)+\receives_n(t)-\serves_n(t)\}$. I.e., server~$n$ completes $\serves_n(t)$ jobs in round~$t$ if it has that many jobs to process.  Otherwise, it completes all $Q_n(t)+\receives_n(t)<\serves_n(t)$ jobs in its possession. 
    As in the case of arrivals, $\serves_n(t)$ is assumed to be stochastically determined. Again, for ease of exposition their distribution is assumed to be the same for all servers.  
    
    \item \T{Communication.} In the fourth phase of a round, dispatchers obtain information about the queue sizes at the servers. 
    We will consider two settings. In the {\em complete information} case, every dispatcher is informed of all queue sizes%
    \footnote{In practice, such information could be gathered in different ways. E.g., by a shared bulletin board or shared memory, as well as by having servers update a central process, who can forward a single message to with $Q(t)$ to each dispatcher.} 
    and so, at the start of round~$t$, it knows $Q(t)$.
    The {\em incomplete information} setting is one in which queue information is communicated over the channels of $E=\setD\times\setS$, without all servers communicating with all dispatchers in every round. 
    
\end{enumerate}

\T{Admissibility.} For a setting as above to be feasible, it must hold that the servers' processing power is sufficient for handling the incoming job requests.
Denoting ${\serves(t) = \sum_{n \in \setS} \serves_n(t)}$, we define the {\em load} to be $\rho=\mathbb{E}[a(0)]/\mathbb{\E}[\serves(0)]$. 
To be admissible, it must hold that $\rho<1$.\footnote{In the appendix, for the purpose of mathematical stability analysis, we also make the standard assumption that the arrival and departure processes admit a finite variance, \ie $\mbox{Var}(a(0)) < \infty$ and $\mbox{Var}(s(0)) < \infty$.} 


\section{Stochastic Coordination}
\label{sec:Silent Coordination:subsec:Stochastic Coordination}\label{sec:Stochastic}

As a first step towards designing an effective multi-dispatcher policy, let us consider the problem in the simpler, single dispatcher, case.  
In round~$t$, this dispatcher has access to the vector $Q(t)$ of queue sizes at the servers, and to the number $a(t)$ of jobs that have been submitted to the system.
If the dispatcher is free to send different jobs to different servers then it will, intuitively, dispatch jobs one by one according to the JSQ principle.
It will send the first job to a queue of shortest length, and then iterate through the jobs, each time sending the next job to a queue of smallest size given the jobs that it has already assigned. Consequently, all queues to which jobs are sent end up with the same sizes (give or take~1). We can view this as being analogous to a process of filling water into a container: Consider a water container with an uneven bottom at heights  that correspond to the histogram 
(rearranged in sorted order) defined by~$Q(t)$. If we should pour a volume of $a(t)$ units of water into the container, then the water level will coincide with the height of the queues to which jobs were dispatched, up to a rounding error due to the fact that water is continuous and jobs are discrete. The largest amount of water would be poured into the deepest column, just as the largest number of jobs would be sent to the shortest queue. 

On a given input $(Q,a)$,  \Cref{alg:WaterLevel}  computes the water level $\wl=\wlf{Q}{a}$ that results from pouring~$a$ units of water into a container with bottom shaped according to~$Q$. The height of each column once the water is poured would be $Q^*\triangleq(Q^*_1,\dots,Q^*_N)$ where
$Q^*_n\triangleq \max\{Q_n, \wl\}$.
We remark that~$Q^*_n$ is not always an integer but might also be a rational number.
If queue lengths are maintained in sorted order, then \wl{} is efficiently computable in $O(\min(N,a))$ time complexity. Measured running times, showing that our policy's runtime scales similarly to JSQ, can be found in Appendix \ref{app:rtmeasuments}.

\begin{algorithm}[ht]
	\footnotesize
    \begin{algorithmic}[1]
     \Function{WaterLevel}{$Q, a$} \Comment{$Q$ is the multiset of queue lengths; $a$ is the total number of arrivals;}
     \\\Comment{$Min$ returns the minimal value in a multiset.\hskip2.87cm$~$ }
         \While{$a > 0$}
        \State $MinSet \gets \set{Q_n \in Q \mid Q_n = Min(Q)}$ \Comment{The set of all minimal queues.}
        \If{$|MinSet| = |Q|$} \Comment{$|\cdot|$ denotes the cardinality of a set.}
            \State \Return $Min(Q) + a/|Q|$ 
        \EndIf
        \State $NextMin \gets Min(Q \backslash MinSet)$
        \State $\delta \gets NextMin - Min(Q)$
        \If{$\delta \cdot |MinSet| < a$}
       	    \State $a \gets a - \delta \cdot |MinSet|$ 
       	    \For{$Q_n \in MinSet$}
       	        \State $Q_n \gets Q_n + \delta$
       	    \EndFor
        \Else
            \State \Return $Min(Q) + a/|MinSet|$
        \EndIf
        \EndWhile
    \EndFunction
    \end{algorithmic}   
	\normalsize
    \caption{Computing the water level.}
    \label{alg:WaterLevel}
\end{algorithm}

In a multi-dispatcher system dispatchers make decisions independently of each other.
It is thus only natural for them to independently optimize their load balancing decisions.
Namely, to dispatch jobs in exactly the same manner as if they were alone in the system.
This is indeed the case in current solutions.
However, in a system where one is not alone, such oblivious behavior of disregarding the others may result in sub-optimal performance \cite{ngynx_po2,youtube_lec,vargaftik2020lsq,zhou2020asymptotically,netflixEdge}.
This occurs when the same server is identified as the best destination by different dispatchers.
These dispatchers then simultaneously forward jobs to this server, causing its queue to grow rapidly, increasing delay times and sometimes even causing the server to drop jobs.
However, the fact that dispatchers make decisions independently does not mean that their decision making protocols must be independently optimized.

In the multi-dispatcher context that we are considering, dispatchers cannot directly coordinate their actions in every given round, since they do not directly communicate with each other. Nevertheless, it is possible to design their protocols in such a way that their actions will be compatible with each other, and will not conflict. The key to doing so is employing {\em randomized protocols}, in which the dispatchers' moves are stochastic.  Indeed, randomization has been a standard tool for symmetry breaking in distributed computing for over four decades \cite{Rabin80,Lehmann81}.  Our goal will be to design probabilistic load balancing protocols that will provide good performance by optimizing the cumulative behavior of the dispatchers. 
This will provide the dispatchers with a silent form of stochastic coordination.


\section{Tidal Water Filling}
\label{sec:Tidal Water Filling}

Focusing on the complete information setting, we assume that each dispatcher~$m$  has access to the number $a^{(m)}=a^{(m)}(t)$ of jobs that it has received in the current round, and to the vector of server queue sizes $Q=Q(t)$ (we shall omit the round number~$t$ when it is clear from context). 
Based on~$Q$ and $a^{(m)}$, it needs to decide where to send each job. 
We seek a solution that will be feasible to compute and amenable to analysis. In particular, we seek a policy that 
(1) is uniform for all dispatchers; given the same~$Q$ and $a^{(m)}=a^{(m')}$, both $m$ and~$m'$ should act in the same way, and in which 
(2) a dispatcher treats all jobs uniformly. 
Hence, the output of $m$'s computation is a vector $\langle p_1,\ldots,p_N \rangle$, where $p_n$ is the probability that any given job will be sent by~$m$ to server~$n$, for every $n\in\setS$. 

We denote by $\brec{n}{m}$ the random variable specifying the number of jobs sent to server~$n$ by dispatcher~$m$. 
The total number of jobs received by~$n$ is $\bar{\receives}_n=\sum_{m\in\setD}\brec{n}{m}$, and its queue size once it receives them is the random variable $\bar{Q}_n\triangleq Q_n+\bar{\receives}_n$.
Finally, we shall denote $\bar{Q}\triangleq\langle\bar{Q}_1,\ldots,\bar{Q}_N\rangle$.

Recall that~$Q$ and the total number of jobs $a=a(t)$ determine a water-filling solution $Q^*=Q^*(Q,a)$ as described in Section~\ref{sec:Stochastic}. Our goal will be to design a policy that  computes the dispatching probabilities~$P$ in such a way that the resulting  queue sizes~$\bar{Q}$ approximate~$Q^*$ as well as possible. More formally, we wish to minimize the $L_2$ distance between~$Q^*$ and~$\bar{Q}$. 
Intuitively, a large distance from the water level \wl$=$\wlf{$Q$}{$a$} induces a large  delay in response times (for a positive difference) or server starvation and lesser resource utilization (negative difference). Since we seek to avoid long delay tails as well as unnecessary server idleness, we consider a large deviation from the \wl{} to be worse than several small ones.
This rules out linear or sub-linear distance measures such as the $L_1$ distance.
On the other hand, giving too much weight to large deviations may miss opportunities to optimize the mean. For example, the $L_\infty$ distance (\ie min-max) addresses only the largest deviation.
We therefore choose to use the $L_2$ distance, since it balances these two desires and is amenable to formal analysis.

We denote the vector of job arrivals at the dispatchers by $\veca=\langle a^{(1)},\ldots,a^{(M)}\rangle$.  
As an interim step, we derive a policy that computes the dispatching probabilities based on~$Q$ and the full vector $\veca$ of jobs that arrive in the round, and not only the allocation $a^{(m)}$ of a single dispatcher~$m$.  We will later discuss how this analysis can be applied to an individual dispatcher's computation. 
Notice that $Q$ and~$\veca$ uniquely determine a (fixed) vector $Q^*$ resulting from water filling~$Q$ with $a = \sum_{m\in\setD} a^{(m)}$ new jobs.
Now, a policy $P(Q,\veca)$ gives rise to the random variable vector $\bar{Q}$ of queue sizes, as described above.

Recall that $\bar{\receives}_n = \bar{Q}_n-Q_n$. Similarly, we denote ${\receives^*_n \triangleq Q^*_n - Q_n}$. Our goal is to minimize 
\begin{equation}\label{eq:l2_error_development}
\begin{aligned}
        \mathbb{E}\lVert Q^*-\bar{Q}\rVert_2^2 ~=~~ 
        &\mathbb{E}\lVert (Q_1^*-\bar{Q}_1,\dots,Q^*_N-\bar{Q}_N)^T\rVert_2^2 ~~= \cr
        &\mathbb{E}\lVert (Q_1+\receives_1^*-Q_1-\bar{\receives}_1,\dots,Q_N+\receives^*_N-Q_N-\bar{\receives}_N)^T\rVert_2^2 ~~= \cr 
        &\mathbb{E}\lVert (\receives_1^*-\bar{\receives}_1,\dots,\receives^*_N-\bar{\receives}_N)^T\rVert_2^2 ~= 
         \sum_{n\in \setN} \mathbb{E}\left[(\receives_n^*-\bar{\receives}_n)^2\right] ~~= \cr
        &\sum_{n\in \setN}{\receives_n^*}^2   - 2\sum_{n\in \setN}\left(\receives_n^*\mathbb{E}\left[\bar{\receives}_n\right] \right)+ \sum_{n\in\setN}\mathbb{E}\left[\bar{\receives}_n^2 \right]
\end{aligned}
\end{equation}
We perform separate analyses for the splittable and for the unsplittable cases. 

\subsection{The Splittable Case}
\label{sec:Tidal Water Filling:subsec:Splittable}
In the splittable case, every job is sent to a server~$n$ with a probability of~$p_n$. This implies, in particular, that the random variable $\bar{\receives}_n^{(m)}$ admits a binomial distribution, that is, $\bar{\receives}_n^{(m)} \sim \Binomial{a^{(m)}}{p_n}$.
Since each decision at each dispatcher is done independently, $\set{\bar{\receives}_n^{(m)} \mid m \in \setM}$ are independent binomial variables with probability $p_n$.
Thus, $\bar{\receives}_n=\sum_{m\in\setM}\bar{\receives}_n^{(m)}$, where
$\bar{\receives}_n \sim \Binomial{\sum_{m\in\setM}a^{(m)}}{p_n} \sim \Binomial{a}{p_n}$.
Hence,
\begin{equation}\label{eq:split_moments}
    \mathbb{E}[\bar{\receives}_n] = a p_n \quad\mbox{and}\quad
    \mathbb{E}[ \bar{\receives}_n^2] = a p_n(1-p_n)+ a^2 {p_n}^2.
\end{equation}
Given $Q$ and $a$ we can rewrite \eqref{eq:l2_error_development} using \eqref{eq:split_moments} as a function of $P = \langle p_1,\ldots,p_N \rangle$:
\begin{equation}\label{eq:l2_error_split}
\begin{aligned}
        f(P) &~=~ 
        \mathbb{E}\lVert Q^*-\bar{Q}\rVert_2^2 ~=~~
        \sum_{n\in\setN}{\receives_n^*}^2   - 2a\sum_{n\in\setN}\receives_n^*p_n+ \sum_{n\in\setN}\left(a p_n - a {p_n}^2 + a^2 {p_n}^2 \right) \cr
        &~=~ \sum_{n\in\setN}{\receives_n^*}^2   - 2a\sum_{n\in\setN}\receives_n^*p_n + a - a\sum_{n\in\setN}{p_n}^2 + a^2\sum_{n\in\setN} {p_n}^2
\end{aligned}
\end{equation}
Now, to simplify the analysis, we first make the observation that for any strictly positive number of arrivals to the system (\ie $a>0$),
\begin{equation}\label{eq:split:simplified}
    \arg\min f(P) ~=~ \arg\min \, (a-1)\sum_{n\in\setN}{p_n}^2 - 2\sum_{n\in\setN} \receives_n^*p_n.
\end{equation}
We thus turn to solve the expression on the right-hand side. As can be seen from \eqref{eq:split:simplified}, for a single arrival to the system (\ie $a=1$), the solution would be to divide the probabilities arbitrarily among all shortest queues. Thus, we next assume that $a>1$.
Recall that we aim to compute a probability assignment~$P$ that optimizes $\arg\min f(P)$. In particular, we have that $\sum_{n\in\setN}p_n = 1$ and $p_n \ge 0 \,\, \forall n\in\setN$. 
The optimization problem to solve in standard form is, 

\begin{equation}\label{eq:split:objective_fuction}
\begin{aligned}
\min_{P} \quad & \tilde{f}(P) = (a-1)\sum_{n\in\setN}{p_n}^2 - 2\sum_{n\in\setN} \receives_n^*p_n\\
\textrm{s.t.} \quad & \sum_{n\in\setN} p_n - 1 = 0, \quad
 -p_n \le 0 \,\, \forall n \in \setN. 
\end{aligned}
\end{equation}
Notice that this is not a linear program, because the objective function is not linear. However, the problem is convex with affine constraints. 
To solve the problem, we employ the Karush-Kuhn-Tucker method (KKT)  \cite{karush1939minima,kuhn1951}. 
The associated Lagrangian function is
\begin{equation}
    \begin{split}
        L(P,\Lambda)& = (a-1)\sum_{n\in\setN}{p_n}^2 - 2\sum_{n\in\setN} \receives_n^*p_n - \sum_{n\in\setN} \Lambda_n p_n + \Lambda_0 (\sum_{n\in\setN} p_n - 1).
    \end{split}
\end{equation}
The respective KKT conditions are
\begin{equation}\label{eq:split:kkt_conditions}
\begin{aligned}
&\frac{\partial L}{\partial p_n} = 2(a-1) p_n - 2\receives_n^* - \Lambda_n + \Lambda_0 = 0 \quad \forall \, n\in\setN & \text{(Stationarity)} \cr
&\sum_{n\in\setN} p_n - 1 = 0 \,\, \text{and } p_n \ge 0 \quad \forall n\in\setN  &\text{(Primal feasibility)}\cr
& \Lambda_n \ge 0 \quad \forall n\in\setN  &\text{(Dual feasibility)}\cr
& p_n \Lambda_n = 0 \quad \forall n\in\setN  & \text{(Complementary slackness)}
\end{aligned}
\end{equation}
By Stationarity 
in 
\eqref{eq:split:kkt_conditions} we obtain that, for any $p_n$, 
\begin{equation}\label{eq:split:positive pi  2}
    p_n = \frac{2\receives_n^* - \Lambda_0 + \Lambda_n}{2(a-1)},
\end{equation}
and adding Complementary slackness from \eqref{eq:split:kkt_conditions} yields that for any $p_n > 0$
we have,
\begin{equation}\label{eq:split:positive pi}
    p_n = \frac{2\receives_n^* - \Lambda_0}{2(a-1)}.
\end{equation}
We can now  substitute for~$p_n$ according to \eqref{eq:split:positive pi} in our objective function  \eqref{eq:split:objective_fuction} to obtain a function of a single variable $\Lambda_0$. This yields,
\begin{equation}\label{eq:split:simple l2 error}
    \tilde{f}(P(\Lambda_0)) = (a{-}1)\sum_{p_n > 0}{\bigg(\frac{2\receives_n^* {-} \Lambda_0}{2(a{-}1)}\bigg)}^2 - 2\sum_{p_n > 0} \receives_n^*\bigg(\frac{2\receives_n^* {-} \Lambda_0}{2(a{-}1)}\bigg) = \frac{\sum_{p_n>0}\Big(\Lambda_0^2 - (2\receives_n^*)^2\Big)}{a-1}.
\end{equation}
Notice that \eqref{eq:split:positive pi} implies that for every $p_n>0$ it holds that $\Lambda_0^2 - (2\receives_n^*) < 0$. 
Hence, every term in the summation of \eqref{eq:split:simple l2 error} is negative.
As a result, lower values of $\Lambda_0$ lead to both lower values of each term and, perhaps, more negative terms.
Clearly, to minimize the objective function, we seek the smallest $\Lambda_0$ that satisfies the KKT conditions given in \eqref{eq:split:kkt_conditions}.
Observe that we can lower bound $\Lambda_0$ by combining \eqref{eq:split:positive pi 2} with the Primal feasibility in  \eqref{eq:split:kkt_conditions} to obtain:
\begin{equation}\label{eq:split:minimal l0 }
    1~=~ \sum_{n\in\setN} p_n ~=~  \sum_{n\in\setN} \frac{2\receives_n^* - \Lambda_0 + \Lambda_n}{2(a-1)} 
    ~\ge~ \sum_{\receives_n^*>0} \frac{2\receives_n^* - \Lambda_0 + \Lambda_n}{2(a-1)}.
\end{equation}
Using $\sum_{n\in\setN}\receives_n^* ~=~ \sum_{\receives_n^*>0}\receives_n^* ~=~ a$, and rearranging \eqref{eq:split:minimal l0 } yields
\begin{equation}\label{eq:split:minimal l0 2}
    2a - \Lambda_0 \sum_{\receives_n^*>0} 1 ~+~ \sum_{\receives_n^*>0} \Lambda_n ~\le~ 2(a-1).
\end{equation}
Thus, due to the Dual feasibility in \eqref{eq:split:kkt_conditions}, we obtain
\begin{equation}\label{eq:split:minimal l0 3}
    \Lambda_0 \ge \frac{2 + \sum_{\receives_n^*>0} \Lambda_n}{\underWL} \ge  \frac{2}{\underWL},\quad \text{where} \quad \underWL ~\triangleq~ \sum_{\receives_n^*\,>\,0} 1.
\end{equation}
Setting $\Lambda_0 = \frac{2}{\underWL}$ and $\Lambda_n = 0$  for all $\receives_n^*>0$ respects the KKT conditions and minimizes the objective function with respect to $\Lambda_0$. Finally, substituting $\frac{2}{\underWL}$ for $\Lambda_0$ in Equation~\eqref{eq:split:positive pi}, we obtain that 
the optimal solution for $a>1$ in the splittable case is
\begin{equation}\label{eq:split:opt_sol}
    p_n ~=~\max \{0, \frac{\receives_n^*-1/\underWL}{a-1} \}~.
\end{equation}

\begin{definition}[Splittable tidal water filling]\label{def:sTWF}
Given $Q=Q(t)$ and $a=a(t)>1$, a stochastic dispatching policy $P(Q,a)$ that, in every round~$t$ sends each job to server~$n\in\setS$ with probability $p_n= \max \{0, \frac{\receives_n^*-1/\underWL}{a-1}\}$, implements tidal water filling (\sTWF) in the splittable setting. 
\end{definition}

Notice that \sTWF\/ depends only on~$a=a(t)$ and~$Q=Q(t)$. It does not depend on the full detail of~$\veca$.
In the context of complete information, $Q$ is available to the dispatcher.
However, $a$ is not.
In order to use \sTWF\/ an individual dispatcher~$m$ must replace $a$ with some estimate.
If $\mathbb{E}[a(0)]$, the expected value of~$a$, is known, it can be used. 
Similarly, if $\mathbb{E}[\serves(0)]$, the total expected completion rate of the servers is known, it may also be used to replace~$a$.
Since, by assumption, dispatcher~$m$ has access to $a^{(m)}(t)$, it can use $M a^{(m)}(t)$ for~$a(t)$. 
This has the nice property that the average of what the dispatchers use equals exactly the total arrivals at that round. That is, $\frac{1}{M}\sum_{m\in\setM} M a^{(m)}(t) = a(t)$.
We will hereafter assume that dispatcher~$m$ estimates~$a(t)$ in this manner.
In Appendix~\ref{app:stabiltiy} we show that the resulting protocol satisfies the desirable {\em strong stability property} for discrete-time queuing systems.%
\footnote{In fact, Appendix~\ref{app:stabiltiy} proves strong stability for all the policies we introduce in this paper including the policies that operate based on partial information (see \Cref{sec:Partial Information}).}
In \cref{sec:Tidal Water Filling:subsec:Example} we shall discuss the properties and the intuitive interpretation of the probabilities used in \sTWF.


\subsection{The Unsplittable Case}
\label{sec:Tidal Water Filling:subsec:Unsplittable}
In the unsplittable case, we again assume that every dispatcher~$m$ knows the vector~$Q(t)$ of queue sizes (complete information).
It differs from the splittable case only in that~$m$ must send all of the jobs that it receives in a given round to a single server. This affects the mathematics of the optimization problem. First of all, knowing the complete vector~$\veca$ of arrivals makes a significant difference in this case. Indeed, given~$Q$ and~$\veca$, computing an optimal job assignment to the servers essentially requires solving an instance of \textsc{Bin-Packing}. 
In \cref{sec:NP-hard} we prove its NP-hardness by a reduction from the \textsc{Partition} problem~\cite{karp1972reducibility}.
More precisely, we show the following.
\begin{restatable}{theorem}{NPhard}
\label{thm:NP-hard}
Given~$Q$ and~$\veca$, it is NP-hard to decide if $\min\mathbb{E}\lVert Q^*-\bar{Q}\rVert_2^2=0$ for a system with two servers.
\end{restatable}
In general, the values $a^{(1)},a^{(2)},\ldots,a^{(M)}$ may be different from each other.
\cref{thm:NP-hard} implies that optimizing for general~$Q$ and~$\veca$ is intractable.
Instead, we optimize the unsplittable problem for the case that $a^{(1)}=a^{(2)}=\cdots=a^{(M)}$, which is consistent with the assumption that $a=Ma^{(m)}$, made in the splittable setting. 
We consider this case as a heuristic means to derive the dispatching probabilities. 
Our experiments in \Cref{sec:eval:complete} show that the resulting policy works well, even when arrivals are governed by {\em i.i.d.}\ Poisson distributions, under which the arrival values $a^{(m)}$ are rarely identical.

We now reformulate the optimization problem in the unsplittable setting for this heuristic case.
The difference from the splittable setting arises following \eqref{eq:l2_error_development} since now $\bar{\receives}_n^{(m)}$ does not admit a Binomial distribution.
Instead, we have
\begin{empheq}[left={\bar{\receives}_n^{(m)} =\empheqlbrace}]{alignat*=2}
    &a^{(m)}, &\text{w.p. } p_n, \\
    &0,       &\text{otherwise}.
\end{empheq}
Namely, a dispatcher $m$ sends all  $a^{(m)}$ of its jobs to server~$n$ with probability $p_n$. 
Therefore, the total number of jobs that server~$n$ receives is $\bar{\receives}_n=\sum_{m\in\setM}\bar{\receives}_n^{(m)}$.
Since $\set{\bar{\receives}_n^{(m)} \mid m \in \setM}$ are {\it i.i.d.}\ random variables, $\bar{\receives}_n$ has the following first and second moments,
\begin{equation}\label{eq:unsplit_moments}
\begin{aligned} 
    \mathbb{E}[\bar{\receives}_n] &= M a^{(m)} p_n, \cr
    \mathbb{E}[ \bar{\receives}_n^2] &= \sum_{m\in\setM}(a^{(m)})^2 p_n  + \hspace{-0.2cm} \sum_{\substack{m,m' \in \setM, \\ m \neq m'}}\hspace{-0.2cm} a^{(m)} \receives^{(m')} {p_n}^2 
    = M (a^{(m)})^2 p_n + M(M-1)(a^{(m)} p_n)^2.
\end{aligned}
\end{equation}
Given $Q$ and $a^{(m)}$ we rewrite \eqref{eq:l2_error_development} substituting the first and second moments according to~\eqref{eq:unsplit_moments}. This yields
\begin{equation}\label{eq:l2_error_unsplit}
\begin{aligned}
       f(P) =& \sum_{n\in\setN}{\receives_n^*}^2 - 2Ma^{(m)}\sum_{n\in\setN}\receives_n^*p_n  
       +\sum_{n\in\setN}\left(M (a^{(m)})^2 p_n {+} M(M{-}1)(a^{(m)} p_n)^2 \right) \cr
        =& \sum_{n\in\setN}{\receives_n^*}^2 - 2Ma^{(m)}\sum_{n\in\setN}\receives_n^*p_n
        +M (a^{(m)})^2\sum_{n\in\setN} p_n  {+}  (M^2{-}M)(a^{(m)})^2\sum_{n\in\setN}{p_n}^2.
\end{aligned}
\end{equation}
Recall that we aim to minimize $f(P)$ under the constraint that~$P$ is a probability assignment.
That is, $\sum_{n\in\setN}p_n = 1$ and $p_n \ge 0 \,\, \forall n\in\setN$.
Observe that, 
\begin{equation}\label{eq:unsplit:simplified}
    \arg\min f(P) = \arg\min \, (M-1)a^{(m)}\sum_{n\in\setN}{p_n}^2 - 2\sum_{n\in\setN} \receives_n^*p_n.
\end{equation}

We proceed to solve for the right hand side. 
As can be seen from \eqref{eq:unsplit:simplified}, for a single-dispatcher system (\ie $M=1$), the solution would be to arbitrarily divide the probabilities among all shortest queues, \ie JSQ.
Thus, we next assume that $M>1$. 
The optimization problem in standard form becomes 
\begin{equation}\label{eq:unsplit:standard form}
\begin{aligned}
\min_{P} \quad & \tilde{f}(P) = (M-1)a^{(m)}\sum_{n\in\setN}{p_n}^2 - 2\sum_{n\in\setN} \receives_n^*p_n\\
\textrm{s.t.} \quad & \sum_{n\in\setN} p_n - 1 = 0, \quad
 -p_n \le 0 \,\, \forall n \in \setN. 
\end{aligned}
\end{equation}
Once more, we employ the KKT method.
However, for the unsplittable case, the solution is more involved.
In particular, identifying the subset of servers $\setU\subseteq\setN$ for which positive probabilities should be assigned, poses a challenge.

The optimization problem in \eqref{eq:unsplit:standard form} is a convex problem with affine constraints.
The resulting Lagrangian function is
\begin{equation}
        L(P,\Lambda) = (M-1)a^{(m)}\sum_{n\in\setN}{p_n}^2 - 2\sum_{n\in\setN} \receives_n^*p_n 
        - \sum_{n\in\setN} \Lambda_n p_n 
        + \Lambda_0 (\sum_{n\in\setN} p_n - 1),
\end{equation}
and the respective KKT conditions are,
\begin{equation}\label{eq:unsplit:kkt_conditions}
\begin{aligned}
&\frac{\partial L}{\partial p_n} = 2(M{-}1)a^{(m)} p_n - 2\receives_n^* - \Lambda_n + \Lambda_0 = 0 \quad \forall \, n\in\setN  & \text{(Stationarity)} \cr
&\sum_{n\in\setN} p_n - 1 = 0 \,\, \text{and } p_n \ge 0 \quad \forall n\in\setN  &\text{(Primal feasibility)}\cr
& \Lambda_n \ge 0 \quad \forall n\in\setN  &\text{(Dual feasibility)}\cr
& p_n \Lambda_n = 0 \quad \forall n\in\setN  & \text{(Complementary slackness)}
\end{aligned}
\end{equation}
Using the Stationarity from \eqref{eq:unsplit:kkt_conditions} we get that for any $p_n$,
\begin{equation}\label{eq:unsplit:positive pi  2}
    p_n = \frac{2\receives_n^* - \Lambda_0 + \Lambda_n}{2(M-1)a^{(m)}},
\end{equation}
and using Complementary slackness \eqref{eq:unsplit:kkt_conditions} yields that for any $p_n > 0$
we have
\begin{equation}\label{eq:unsplit:positive pi}
    p_n = \frac{2\receives_n^* - \Lambda_0}{2(M-1)a^{(m)}}.
\end{equation}
We can use \eqref{eq:unsplit:positive pi} in our objective function in \eqref{eq:unsplit:standard form} to obtain a function of a single variable $\Lambda_0$. This yields,
\begin{equation}\label{eq:unsplit:simple l2 error}
    \tilde{f}(P(\Lambda_0)) =~ (M{-}1)a^{(m)}\sum_{p_n > 0}{\bigg(\frac{2\receives_n^* - \Lambda_0}{2(M{-}1)a^{(m)}}\bigg)}^2 
     - 2\sum_{p_n > 0} \receives_n^*\bigg(\frac{2\receives_n^* - \Lambda_0}{2(M{-}1)a^{(m)}}\bigg).
\end{equation}
Rearranging \eqref{eq:unsplit:simple l2 error} we get,
\begin{equation}\label{eq:unsplit:simpler l2 error}
    \tilde{f}(P(\Lambda_0)) = \frac{1}{(M{-}1)a^{(m)}}\sum_{p_n>0}\Big(\frac{\Lambda_0^2}{4} - {\receives_n^*}^2\Big).
\end{equation}
Similarly to the splittable case, in order to minimize the objective function, we need to find the smallest $\Lambda_0$ that respects the KKT conditions.
For this purpose we define a set of servers $\setU\subseteq \setN$ that would be strictly below the water level if the total arrivals were $(M{-}1)a^{(m)}$ jobs instead of $Ma^{(m)}$.%
\footnote{$\setU$ is not directly derived from the equations. Identifying it is partially based on an intelligent guess.}
That is, \mbox{$\setU=\{n \mid Q_n <  \wlf{Q}{(M{-}1)a^{(m)}}  \}$}, where  \textsc{WaterLevel} is given by \cref{alg:WaterLevel}.
We now use \eqref{eq:unsplit:positive pi 2} and the Primal feasibility \eqref{eq:unsplit:kkt_conditions} to obtain,
\begin{equation}\label{eq:unsplit:minimal l0 }
    1=
    \sum_{n\in\setN} p_n =
    \sum_{n\in\setN} \frac{2\receives_n^* - \Lambda_0 + \Lambda_n}{2(M{-}1)a^{(m)}} \ge
    \sum_{n\in \setU} \frac{2\receives_n^* - \Lambda_0 + \Lambda_n}{2(M{-}1)a^{(m)}}.
\end{equation}
Rearranging \eqref{eq:unsplit:minimal l0 } and using the Dual feasibility from \eqref{eq:unsplit:kkt_conditions} yields,
\begin{equation}\label{eq:unsplit:minimal l0 2}
    2(M{-}1)a^{(m)} \ge
    2 \sum_{n\in \setU} \receives_n^* - \Lambda_0 \sum_{n\in \setU} 1 + \sum_{n\in \setU} \Lambda_n 
    \ge 2 \sum_{n\in \setU} \receives_n^*  - \Lambda_0 |\setU|.
\end{equation}
Thus, since $Ma^{(m)} = \sum_{n\in\setN} \receives_n^*$ and $\sum_{n\in\setN} \receives_n^* - \sum_{n\in \setU} \receives_n^* = \sum_{n\notin \setU} \receives_n^*$, we get
\begin{equation}\label{eq:unsplit:minimal l0 3}
    \Lambda_0 \ge 
    \frac{2}{|\setU|} \left(a^{(m)}-\sum_{n\notin \setU} \receives_n^*\right).
\end{equation}
Next, we prove that setting $\Lambda_0$ on the above lower bound, \ie $\Lambda_0 =  \frac{2}{|\setU|} \left(a^{(m)}-\sum_{n\notin \setU} \receives_n^*\right)$, respects the KKT conditions.
We start with the Dual feasibility condition, $\Lambda_0 \ge 0$.  
Recall that $M\ge 2$, and by definition of $\setU$ it holds that $(M{-}1)a^{(m)}\le \sum_{n\in \setU}\receives_n^* \le M a^{(m)}$, thus,
\begin{equation*}
    M a^{(m)} = \sum_{n\in\setN} \receives_n^* = \sum_{n\in \setU} \receives_n^* +  \sum_{n\notin \setU} \receives_n^*
\text{ and } 
    0\le  \sum_{n\notin \setU} \receives_n^* \le a^{(m)}.
\end{equation*}
Therefore,
\begin{equation*}
    \left(a^{(m)} - \sum_{n\notin \setU} \receives_n^*\right) \ge 0
\text{ and, } 
    \Lambda_0 =  \frac{2}{|\setU|} \left(a^{(m)}-\sum_{n\notin \setU} \receives_n^* \right) \ge 0~.
\end{equation*}
This shows Dual feasibility.

Plugging \mbox{$\Lambda_0 =  \frac{2}{|\setU|} \left(a^{(m)}-\sum_{n\notin \setU} \receives_n^* \right)$} into \eqref{eq:unsplit:positive pi} yields
\begin{equation}\label{eq:unsplittableWfProbabilities}
  p_n =\max \left\{ 0,\frac{\receives_n^* - (a^{(m)}-\sum_{n' \notin \setU} \receives_{n'}^*) / |\setU| } {(M{-}1)a^{(m)}} \right\}~.  
\end{equation}
It only remains to show that these $p_n$'s also satisfy the Primal feasibility condition from~\eqref{eq:unsplit:kkt_conditions}.
We first show that for each \mbox{$n\in \setU$} it holds that $p_n>0$, whereas for $n\notin \setU$ it holds that~$p_n = 0$.
\begin{lemma}\label{lem:unsplittable:U}
    For each $n\in \setN$, it holds that $p_n>0$ if and only if $n\in \setU$.
\end{lemma}
\begin{proof}
    Let $b_n^* = \max \{ 0,\wlf{Q}{(M{-}1)a^{(m)}} {-} Q_n\}$.
    By definition of~$\setU$, it holds that $b_n^*>0$ if and only if $n\in \setU$.
    Note that $Q_n + b_n^*$ is exactly $\wlf{Q}{(M{-}1)a^{(m)}}$ for $n\in \setU$ and is simply $Q_n$ for $n\notin \setU$.
    Thus, for each $n\in \setU$ we obtain $\receives_n^* = b_n^* + x$ where
    \begin{equation*}
    \begin{split}
    x = &\wlf{Q}{Ma^{(m)}} 
         - \wlf{Q}{(M{-}1)a^{(m)}} 
    =  \frac{(a^{(m)}-\sum_{n' \notin \setU} \receives_{n'}^*)}{|\setU|},
    \end{split}
    \end{equation*}
    and thus $p_n > 0$ if $n\in \setU$.
    
    On the other hand, if $n \notin \setU$ we obtain $\receives_n^* = b_n^* + x_n$ where $ b_n^*=0$ and $x_n \le x$ since,
    $$x_n = \max \{0, \frac{\wlf{Q}{M a^{(m)}}-Q_n}{|\setU|}\},$$
    and,
    $$Q_n \ge \wlf{Q}{(M{-}1)a^{(m)}}.$$
    This yields
    $$\receives_n^* = x_n \le x = \frac{(a^{(m)}-\sum_{n' \notin \setU} \receives_{n'}^*)}{|\setU|}, $$ and thus $p_n = 0$ if $n \notin \setU$. This concludes the proof.
\end{proof}
According to \cref{lem:unsplittable:U} we have that
\begin{equation*}
\begin{split}
     \sum_{n\in\setN} p_n = & \sum_{n\in \setU} p_n 
     = \sum_{n\in \setU} \left( \frac{\receives_n^* - (a^{(m)}-\sum_{n' \notin \setU} \receives_{n'}^*) / |\setU|} {(M{-}1)a^{(m)}} \right)
    \cr
    = & \frac{1}{(M{-}1)a^{(m)}} \left(
    \sum_{n\in \setU}  \receives_n^* 
    - \sum_{n\in \setU} \left(  (a^{(m)}-\sum_{n' \notin \setU} \receives_{n'}^*) / |\setU|  \right)
    \right) \cr
     = & \frac{1}{(M{-}1)a^{(m)}} \left(
    \sum_{n\in \setU}  \receives_n^* 
    - |\setU| \frac{a^{(m)} - \sum_{n' \notin \setU} \receives_{n'}^* }{|\setU|}
    \right) \cr
    = & \frac{\sum_{n\in\setN} \receives_n^* - a^{(m)}}{(M{-}1)a^{(m)}}
    = \frac{M a^{(m)} - a^{(m)}}{(M{-}1)a^{(m)}} = 1,
\end{split}
\end{equation*}
and all of the KKT conditions are satisfied.\\

The solution for the optimization problem in \eqref{eq:unsplit:standard form} yields the following notion of tidal water filling for the unsplittable case:
\begin{definition}[Unsplittable tidal water filling]\label{def:uTWF}
Given $Q=Q(t)$ and $a^{(m)}=a^{(m)}(t)$, let $\setU=\{n \mid Q_n <  \wlf{Q}{(M{-}1)a^{(m)}}  \}$.
An individual stochastic dispatching policy $P(Q,a^{(m)})$ for dispatcher~$m$ that, in every round~$t$ sends all its jobs to server~$n\in\setS$ with probability 
\begin{equation}\label{eq:unsplit:opt_sol}
  p_n =\max \left\{ 0,\frac{\receives_n^* - (a^{(m)}-\sum_{n' \notin \setU} \receives_{n'}^*) / |\setU| } {(M{-}1)a^{(m)}} \right\}  
\end{equation}

 \noindent
 implements tidal water filling (u\TWLs) in the unsplittable setting. 
\end{definition}
Observe that, when at most one job arrives at each dispatcher, \ie~$a^{(m)}\le 1$ for all~$m\in\setM$, there is no difference between the splittable and unsplittable problems. Indeed, in this case  Definition~\ref{def:uTWF} and Definition~\ref{def:sTWF} coincide.

\subsection{\TWLs{} vs Water Filling in Expectation}
\label{sec:Tidal Water Filling:subsec:Example}
Recall that we aim to approximate water filling (\ie $Q^*$). 
Consider the policy by which every dispatcher sends a job to each server~$n$ with a probability proportional to the amount of ``water'' it would receive in the pure water-filling solution. 
More formally, we define the Water Filling in Expectation policy (\WFiE) to assign probabilities $P(Q,a)=\langle p_1,\ldots,p_N \rangle$, where for every $n$ we have 
\begin{equation}\label{eq:WFiE}
   p_n = \frac{\receives_n^* } {a}~.
\end{equation}
The expected length of each server~$n$'s queue under \WFiE{} is precisely~$Q^*_n$.
Tidal water filling, however, does not produce the pure water-filling solution in expectation. 
Nevertheless, the math does not lie, and \TWLs{} improves on \WFiE.
We use the following two examples to demonstrate that \WFiE{} is suboptimal, and to provide intuition for why \TWLs{} is better. 

\T{Example 1.}
 \Cref{fig:TWFexample1} depicts a system with $N=2$ servers. At the beginning of the round, server~$n_1$ has a single job in its queue, and server~$n_2$ is idle (its queue is empty). 
There are~$M=2$ dispatchers, each of which receives a single job to dispatch. 
In this scenario, a dispatcher that uses JSQ will send its job to~$n_2$; a dispatcher that uses \WFiE{} will send its job to~$n_1$ with probability~$p_{n_1}=1/4$ and to~$n_2$ with probability~$p_{n_2}=3/4$; a dispatcher that uses \TWLs{} (since each dispatcher has a single job to dispatch in this scenario, \sTWF{} and \uTWF{} coincide) will send its job to~$n_2$ with probability~$p_{n_2}=1$.
The resulting {\em expected lengths} of the queues are depicted in~\Cref{fig:TWFexample:jsq1} for JSQ, in~\Cref{fig:TWFexample:wf1} for \WFiE, and in~\Cref{fig:TWFexample:jsq1} for \TWLs.
Observe that both JSQ and \TWLs{} guarantee the favorable solution in which the longest queue has size~2, while in \WFiE{} there is a non-negligible probability of $1/16$ that both jobs will be forwarded to~$n_1$, creating a queue of size~3.


\begin{figure}[t]
\centering
\begin{subfigure}{.2\textwidth}
  \centering
  \includegraphics[width=0.8\linewidth]{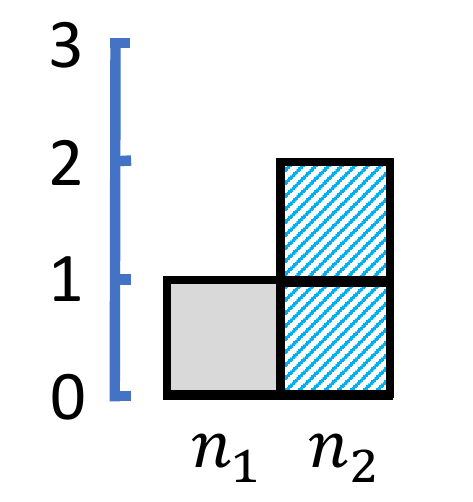}  
  \caption{\footnotesize JSQ.}
  \label{fig:TWFexample:jsq1}
\end{subfigure}
\quad\quad\quad\quad
\begin{subfigure}{.2\textwidth}
  \centering
  \includegraphics[width=0.8\linewidth]{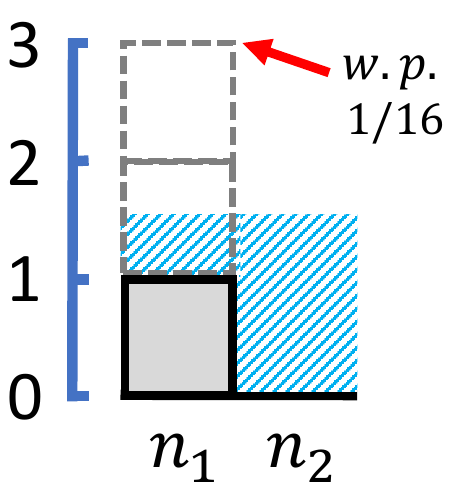}  
  \caption{\footnotesize WFiE.}
  \label{fig:TWFexample:wf1}
\end{subfigure}
\quad\quad\quad\quad
\begin{subfigure}{.2\textwidth}
  \centering
  \includegraphics[width=0.8\linewidth]{waterFillingExample/ex1_jsq_TWF}  
  \caption{\footnotesize \TWLs.}
  \label{fig:TWFexample:twf1}
\end{subfigure}
\vspace{-2mm}
\caption{A scenario with 2 dispatchers each of which receives a single job. Illustrating the expected arrivals at each queue for JSQ, WFiE and \TWLs{} (\Cref{fig:TWFexample:jsq1,fig:TWFexample:wf1,fig:TWFexample:twf1} respectively).}
\label{fig:TWFexample1}
\end{figure}


\begin{figure}[t]
\centering
\begin{subfigure}{.2\textwidth}
  \centering
  \includegraphics[width=0.8\linewidth]{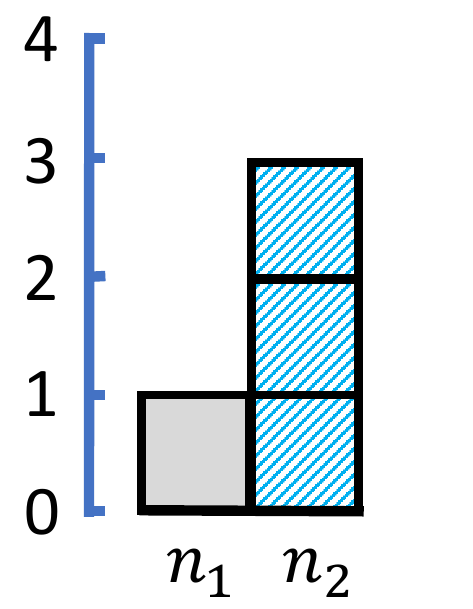}  
  \caption{\footnotesize JSQ.}
  \label{fig:TWFexample:jsq2}
\end{subfigure}
\quad\quad\quad\quad
\begin{subfigure}{.2\textwidth}
  \centering
  \includegraphics[width=0.8\linewidth]{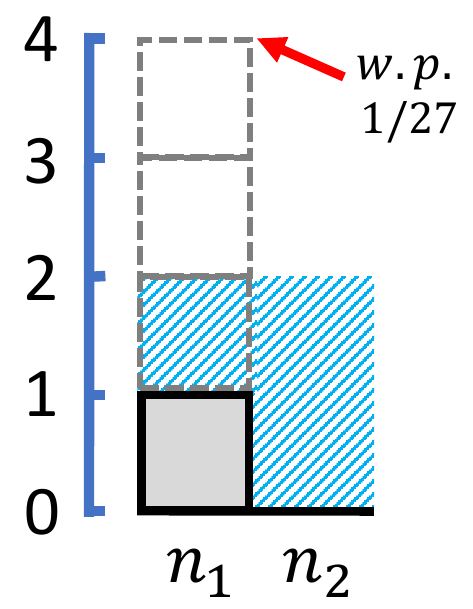}  
  \caption{\footnotesize WFiE.}
  \label{fig:TWFexample:wf2}
\end{subfigure}
\quad\quad\quad\quad
\begin{subfigure}{.2\textwidth}
  \centering
  \includegraphics[width=0.8\linewidth]{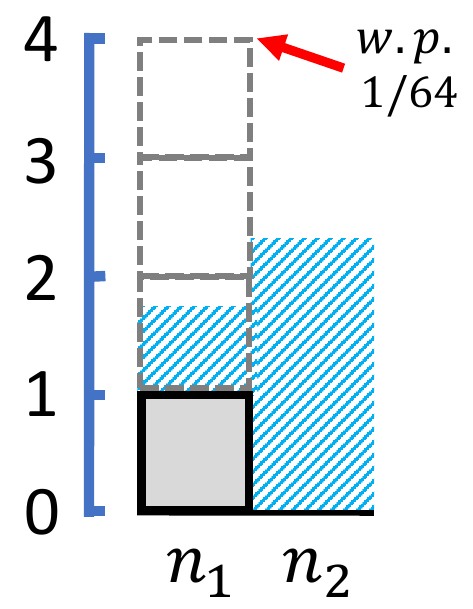}  
  \caption{\footnotesize \TWLs.}
  \label{fig:TWFexample:twf2}
\end{subfigure}
\vspace{-2mm}
\caption{A scenario with 3 dispatchers each of which receives a single job. Illustrating the expected arrivals at each queue for JSQ, WFiE and \TWLs{} (\Cref{fig:TWFexample:jsq2,fig:TWFexample:twf2,fig:TWFexample:wf2} respectively).}
\vspace{-3mm}
\label{fig:TWFexample2}
\end{figure}


\T{Example 2.}
\Cref{fig:TWFexample2} illustrates a system with the same two servers and the same initial state, but with $M=3$ dispatchers. 
Each of the three dispatchers receives a single job to dispatch. 
In this scenario, a dispatcher that uses JSQ will again send its job to~$n_2$; a dispatcher that uses \WFiE{} will send its job to~$n_1$ with probability~$p_{n_1}=1/3$ and to~$n_2$ with probability~$p_{n_2}=2/3$; a dispatcher that uses \TWLs{} will send its job to~$n_1$ with probability~$p_{n_1}=1/4$ and to~$n_2$ with probability~$p_{n_2}=3/4$.
The resulting {\em expected lengths} of the queues are depicted in~\Cref{fig:TWFexample:jsq1} for JSQ, in \Cref{fig:TWFexample:wf1} for \WFiE, and in \Cref{fig:TWFexample:jsq1} for \TWLs.
In this case, JSQ results in herding towards $n_2$, creating a queue of size~3 there.
\WFiE{} results is a probability of~$1/27$ for ending with 4 jobs queuing at~$n_1$, while~$n_2$ remains idle.
Observe that \TWLs{} reduces the probability of this worst-case allocation from~$1/27$ to~$1/64$.
This is precisely the advantage that \TWLs{} provides over \WFiE{} in general. We note that  \TWLs{} also provides, with high probability, a favorable allocation in comparison to JSQ.
In the second example, for instance, JSQ will produce a better outcome than \TWLs{} with probability $1.56\%=1/64$, while \TWLs{} will be better than JSQ with probability $42.2\%=27/64$. 

Roughly speaking, our \TWLs{} policies have a greater bias towards short queues than \WFiE{} does. As illustrated by the examples this, in turn, reduces the probability that queues will grow excessively long, and reduces the probability for short queues to become idle.


\begin{figure}[t]
    \centering
    \begin{subfigure}{\textwidth}
      \centering
      \caption{Splittable dispatching.}
    \includegraphics[width=0.95\linewidth]{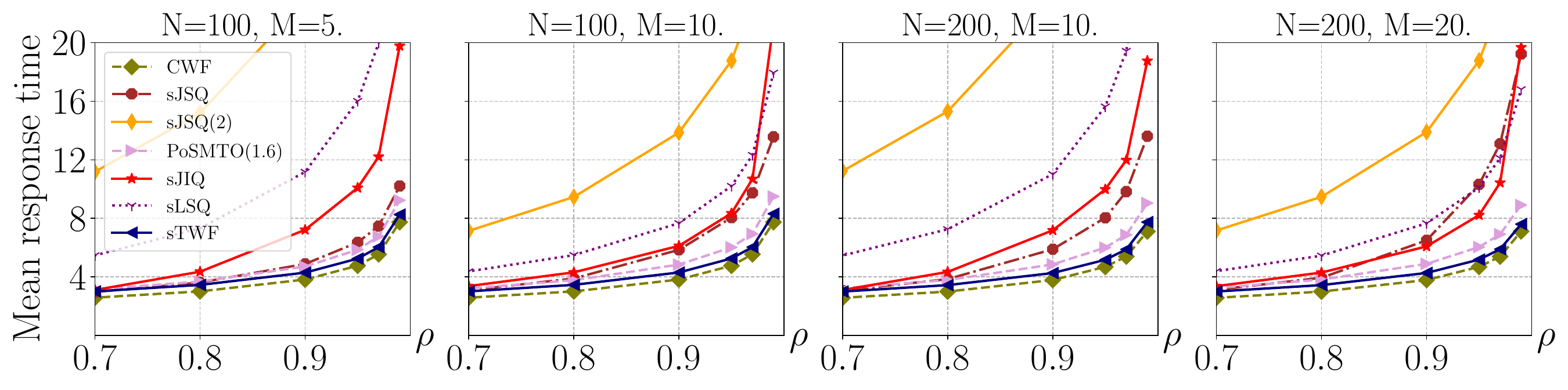}
     \label{fig:evaluation:complete:loadsweep:split}
    \end{subfigure}
    \begin{subfigure}{\textwidth}
      \centering
      \caption{Unsplittable dispatching.}
    \includegraphics[width=0.95\linewidth]{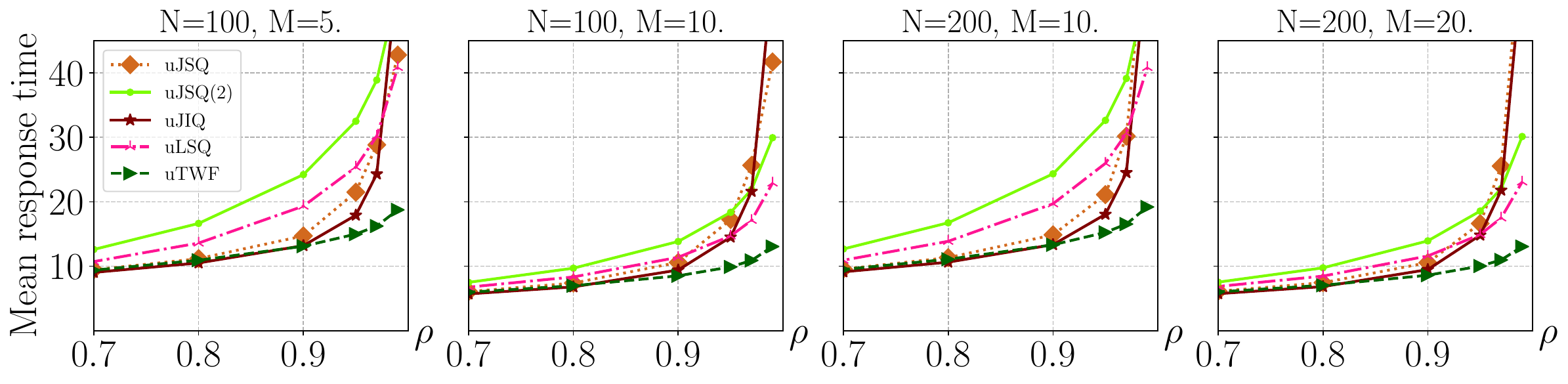}  
     \label{fig:evaluation:complete:loadsweep:unsplit}
    \end{subfigure}
    \caption{Average job response time as a function of the load over four different systems.
     The $x$-axis represents load $\rho$. The $y$-axis represents the average response time.}
    \label{fig:evaluation:complete:loadsweep}
\end{figure}

\subsection{Evaluation}\label{sec:eval:complete}

We conducted an empirical study of our \TWLs{} policies via simulations. 
In all of the simulations, at round~$t$ a dispatcher~$m$ has access only to~$a^{(m)}(t)$. Recall that in both \sTWF\ and \uTWF\ dispatcher~$m$ uses~$M\!\cdot\! a^{(m)}(t)$ as an estimate for $a(t)$.

\T{Arrivals and departures.} Our modeling of the arrival and departure processes follows standard practice (\eg \cite{vargaftik2020lsq,lu2011join,mitzenmacher2002load,van2019hyper,anselmi2020power}).
In each round, we set $a^{(m)}(t) \sim Poisson(\lambda)$ for each dispatcher $m \in \setM$, and $s^{(n)}(t) \sim Geometric(\mu)$ for each server $n \in \setN$.  Therefore, the load on the system is $\rho \triangleq M \lambda / (N \frac{\mu}{1-\mu})$.

\T{Dispatching policies.} We compare our policies to JSQ, the Power-of-two-choices denoted by JSQ(2), JIQ and the recently proposed LSQ.%
\footnote{Specifically, we use LSQ-Sample(2). See~\cite{vargaftik2020lsq} for details.}
For the splittable case, we also compare against: (1) ``the power of slightly more than a single choice" (PoSMTO($d$)) proposed in \cite{ying2017power}\footnote{We calibrated the parameter $d$ using the guidelines in \cite{ying2017power} and found $d=1.6$ as the sweet spot. Namely, the number of samples a dispatcher performs each round is not constant but equals to the size of the arrived batch of jobs multiplied by 1.6.}; (2) a centralized JSQ policy (i.e., all arrivals go through a single dispatcher, which results in a centralized water-filling (CWF) effect). The comparison against this policy can be seen as the ``price of districution''.
We use a prefix of $s$ and $u$ to denote the splittable and the unsplittable case. For the splittable case, similarly to s\TWLs, other policies are also allowed to split jobs for parallel processing.
Namely, in splittable JSQ (sJSQ), each dispatcher performs water-filling considering only its own jobs (i.e., disregarding possible arrivals of jobs to the same servers from other dispatchers). Similarly, this is the case for other policies. For example, splittable JIQ (sJIQ) splits the jobs among the idle queues equally (with random tie breaks). If there are no idle queues, each job is sent to a randomly selected server.

\T{Setup.} We consider systems with different numbers of servers~$N$ and dispatchers~$M$.
For each system, we run a set of simulations. Each simulation is configured with a different load and lasts for $10^5$ time slots (rounds).

\T{Results.}
\Cref{fig:evaluation:complete:loadsweep} shows the mean job response time ($y$-axis) across the different loads ($x$-axis) for different systems.
It is notable that \sTWF\/ outperforms all other policies in the splittable case across all systems, and \uTWF\/ does the same in the unsplittable case. 
Moreover, as the load increases the gap between the TWF policies and the rest grows.

As mentioned in~\cref{sec:Intro} and in~\cite{dean2013tail,nishtala2017hipster,lu2011join,schurman2009user}, tail distributions play a crucial role in the parallel-server setting. We proceed to  measure the tail distributions under various system parameters  under the challenging scenario of a high load of  ($\rho=0.99$).
This is depicted in \Cref{fig:evaluation:complete:ccdf} using the complementary cumulative distribution function (CCDF), which shows for each response time ($x$-axis, denoted by $\tau$), what is the fraction of jobs that surpass it ($y$-axis). 
In summary, in the complete information case, our simulations show that the TWF policies consistently outperform the previous approaches when the load is high. Further, TWF has a considerably lower ``price of distribution'', in the splittable case, compared to other policies as evident by the comparison to CWF.

\T{Additional results.} Both the distributions that are chosen to model the arrivals processes and the size of the system clearly impact the results. Accordingly, we have also conducted simulations with a heavy tail distribution (Log-normal) and a larger system (500 servers). For these simulations, detailed in Appendices \ref{app:htarrivalseval} and \ref{app:largersystemseval}, we observe that the gap is still significantly in favor of TWF.


\begin{figure}[t]
    \centering
    \begin{subfigure}{\textwidth}
      \centering
      \caption{Splittable dispatching.}
    \includegraphics[width=0.95\linewidth]{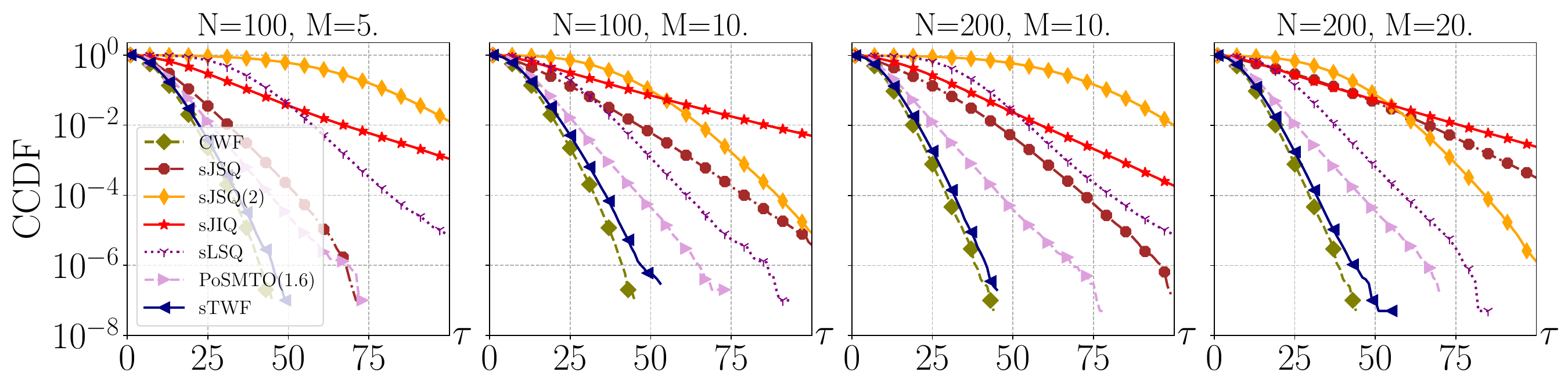}  
     \label{fig:evaluation:complete:ccdf:split}
    \end{subfigure}
    \begin{subfigure}{\textwidth}
      \centering
      \caption{Unsplittable dispatching.}
    \includegraphics[width=0.95\linewidth]{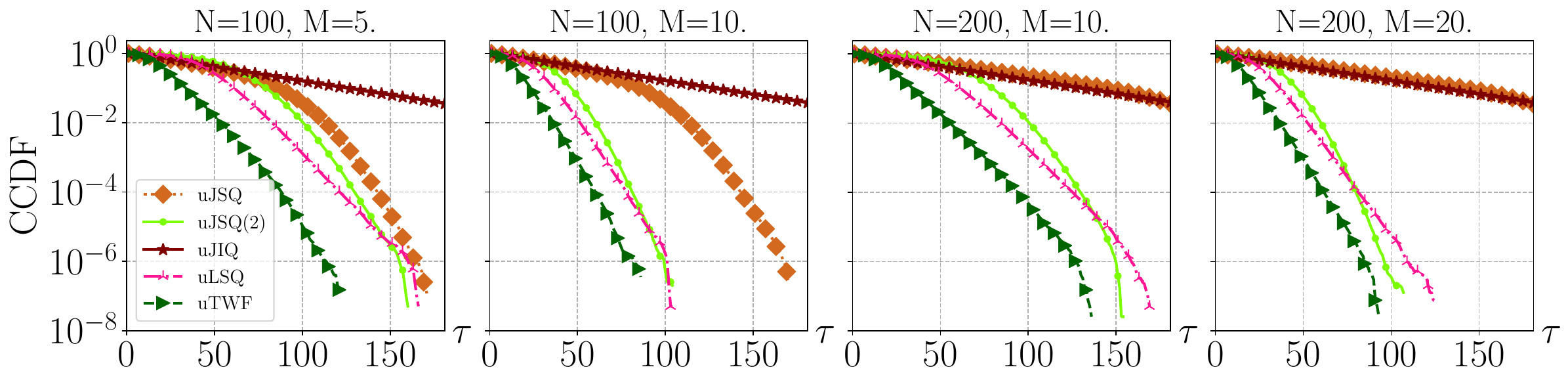}  
     \label{fig:evaluation:complete:ccdf:unsplit}
    \end{subfigure}
    \caption{Response-time tail distribution for four different systems at high load ($\rho=0.99$). The $x$-axis represents the response time (denoted by $\tau$). The $y$-axis represents the CCDF.}
    \label{fig:evaluation:complete:ccdf}
\end{figure}




\section{Enhancing Performance for Partial Information}
\label{sec:Partial Information}
In this section we relax the requirement that dispatchers have complete and accurate information regarding~$Q$.
That is, we now consider situations in which a dispatcher does not know the exact state of all servers.
In line with recent work \cite{vargaftik2020lsq,zhou2020asymptotically,anselmi2020power,van2019hyper}, we consider a system where each dispatcher keeps a local array representing the servers' queue lengths, which may contain inaccurate (e.g., outdated) information.
\begin{wrapfigure}{L}{0.4\textwidth}
\centering
\includegraphics[width=0.4\textwidth]{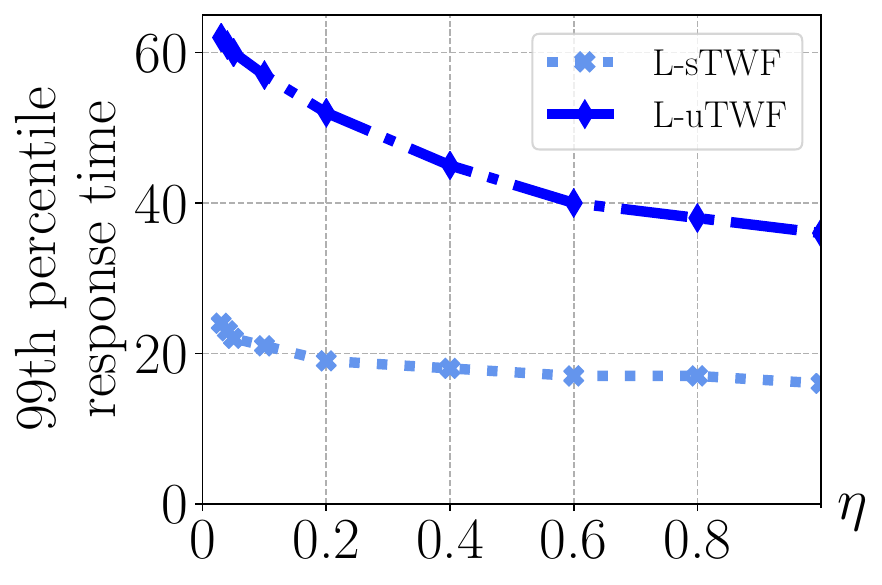}
\caption{Response times as a function of~$\eta$ at a high load \mbox{($\rho=0.99$)} for a system with $(N,M)=(100,10)$.}
\label{fig:evaluation:monotomicity}
\end{wrapfigure}
Communication is used to update array entries in the following manner:
At the end of each round (\ie in the fourth phase), a dispatcher establishes communication links with a fraction $\eta\le 1$ of the servers, which are chosen uniformly at random. 
The corresponding entries in the dispatcher's local state are then updated with these servers' queue lengths.
Additionally, a dispatcher that establishes a communication link to send jobs to server~$n$ during round~$t$ learns $Q_n(t)$.
Notice that $\eta=1$ corresponds to the complete information assumption; we use $\eta$ as a parameter designating the extent of partial information available to the dispatchers. 
A dispatcher that uses \uTWF{} based on its local array is said to implement \emph{Local} \uTWF{} (L-\uTWF).  
\emph{Local} \sTWF{} is defined in the same manner.
\Cref{fig:evaluation:monotomicity} illustrates the simulations results. It shows that the response time improves monotonically as~$\eta$ increases. 

This motivates us to increase the available information to the dispatchers.
We attempt to do so without increasing the number of links that a dispatcher establishes.
This is of interest since in many system the cost of communication lies mainly in the connection establishment rather than in the size of its content~\cite{murray2012large,moon2020acceltcp}.
To that end, we keep track of queue-size information at the servers, in a local array of size~$N$. 
A server updates its local array based on its own queue length and information that it receives from dispatchers with which it has connections.
Whenever a communication link between a dispatcher and a server is established, they merge their arrays.
This is obtained by assigning time stamps to array entries, and maintaining the most recent information per entry upon the merge (cf.~\cite{Lamport78}).%
\footnote{To the best of our knowledge, maintaining queue size information in this manner at the servers has not been done before in the parallel server model.}

To test the effectiveness of the above scheme, we conduct an experiment comparing our protocols with the state-of-the-art LSQ{-}Sample of \cite{vargaftik2020lsq,zhou2020asymptotically} for different values of~$\eta$, i.e., LSQ{-}Sample($\eta N$). 
We denote by \uTWF$^{ts}$ the unsplittable policy from~\Cref{def:uTWF} based on local arrays at both dispatchers and servers with time stamps.
\Cref{fig:evaluation:advancedCom} shows how the performance of \uTWF$^{ts}$ improves on that of L-\uTWF\/ for given values of~$\eta$. 
The figure also illustrates that even the simpler L-\uTWF{} policy is competitive with LSQ already at $\eta=0.1$.
As~$\eta$ grows (and with it the queue information improves), our protocols perform better, while LSQ's performance degrades due to increased herding.

\begin{figure}[h]
    \centering
    \includegraphics[width=0.95\linewidth]{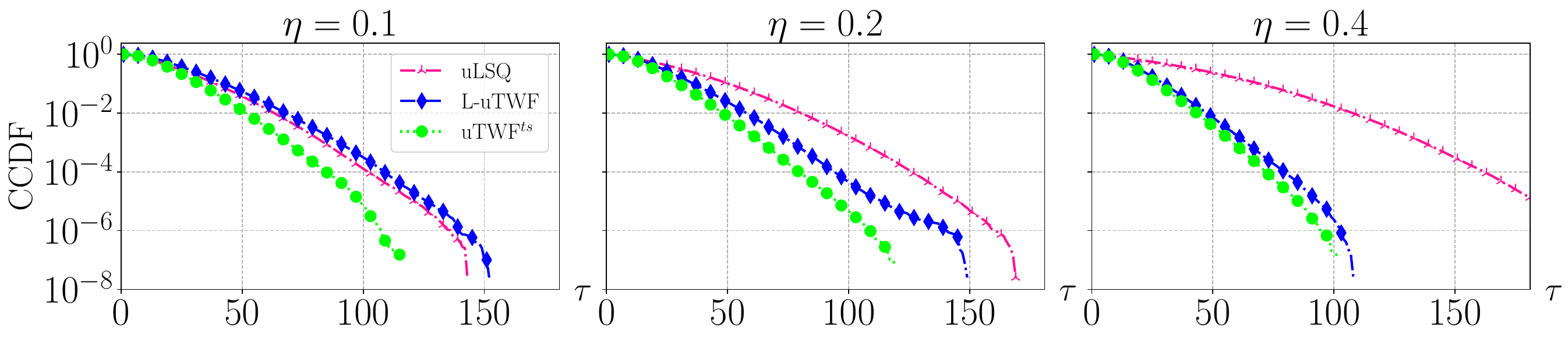}
    \caption{The effect of a distributed communication protocol on performance in a system with $(N,M)=(200,20)$. Measuring the response time tail distribution (\ie CCDF) at high load ($\rho=0.99$).}
     \label{fig:evaluation:advancedCom}
\end{figure}


\section{Discussion}
\label{sec:Disc}
This work has demonstrated that, contrary to popular belief, queue-size information can be judiciously used to improve the quality of load balancing.  
In particular, we provided new policies that avoid herding and outperform all previous solutions for the case of complete~information.

We have made a step forward in understanding the power of stochastic coordination in the load balancing arena. 
There are many additional aspects of stochastic coordination that should be explored.
Could the same conceptual design be used more generally, i.e., by loosening the homogeneity or the connectivity assumptions?
In the partial information setting we note that when queue size information is sparse, the \TWLs's advantages do not come into play, and its performance is not better and may be even worse than that of previous load-balancing policies.

Moreover, for many systems (e.g., wireless, sensors, peer-2-peer, etc.), it may be the case that obtaining complete or even partial information is too prohibitive, and one has to rely on sparse communication. For such a case, policies such as $JSQ(d)$ and $LSQ$ may offer better performance since, unlike TWF, they can operate in such settings (e.g., two samples per round). 

One can view load balancing in the multi-dispatcher parallel server model as a natural question to explore using distributed systems tools and techniques. 
Our analysis in \Cref{sec:Partial Information}, for example, made use of time-stamping and flooding to improve the load balancing performance when information is partial.
We state as an open problem how the advantages of previous approaches can be combined with those of \TWLs{} to obtain a policy that would make the best use of information across the whole spectrum of possibilities. 

In another vein, it would be interesting to investigate how information about the distributions governing a multi-dispatcher systems can be obtained, and how they can be used to improve load-balancing performance. 
Can they provide good estimates for the \TWLs{} policies, and if so, how much benefit can they bring? 
Much is clearly left to be done.







\bibliography{references}

\begin{thebibliography}{10}

\bibitem{adler1998parallel}
Micah Adler, Soumen Chakrabarti, Michael Mitzenmacher, and Lars Rasmussen.
\newblock Parallel randomized load balancing.
\newblock {\em Random Structures \& Algorithms}, 13(2):159--188, 1998.

\bibitem{anselmi2020power}
Jonatha Anselmi and Francois Dufour.
\newblock Power-of-d-choices with memory: Fluid limit and optimality.
\newblock {\em Mathematics of Operations Research}, 2020.

\bibitem{atar2020persistent}
Rami Atar, Isaac Keslassy, Gal Mendelson, Ariel Orda, and Shay Vargaftik.
\newblock Persistent-idle load-distribution.
\newblock {\em Stochastic Systems}, 10(2):152--169, 2020.

\bibitem{atar2021persistent}
Rami Atar, Isaac Keslassy, Gal Mendelson, Ariel Orda, and Shay Vargaftik.
\newblock On the persistent-idle load distribution policy under batch arrivals
  and random service capacity.
\newblock {\em arXiv preprint arXiv:2103.12140}, 2021.

\bibitem{dean2013tail}
Jeffrey Dean and Luiz~Andr{\'e} Barroso.
\newblock The tail at scale.
\newblock {\em Communications of the ACM}, 56(2):74--80, 2013.

\bibitem{eisenbud2016maglev}
Daniel~E Eisenbud, Cheng Yi, Carlo Contavalli, Cody Smith, Roman Kononov, Eric
  Mann-Hielscher, Ardas Cilingiroglu, Bin Cheyney, Wentao Shang, and
  Jinnah~Dylan Hosein.
\newblock Maglev: A fast and reliable software network load balancer.
\newblock In {\em 13th USENIX Symposium on Networked Systems Design and
  Implementation (NSDI)}, pages 523--535, 2016.

\bibitem{eryilmaz2012asymptotically}
Atilla Eryilmaz and Rayadurgam Srikant.
\newblock Asymptotically tight steady-state queue length bounds implied by
  drift conditions.
\newblock {\em Queueing Systems}, 72(3-4):311--359, 2012.

\bibitem{gandhi2014duet}
Rohan Gandhi, Hongqiang~Harry Liu, Y~Charlie Hu, Guohan Lu, Jitendra Padhye,
  Lihua Yuan, and Ming Zhang.
\newblock Duet: Cloud scale load balancing with hardware and software.
\newblock {\em ACM SIGCOMM Computer Communication Review}, 44(4):27--38, 2014.

\bibitem{ngynx_po2}
Owen Garrett.
\newblock {NGINX and the ``Power of Two Choices'' Load-Balancing Algorithm.}
\newblock
  \url{https://www.nginx.com/blog/nginx-power-of-two-choices-load-balancing-algorithm},
  published on November 12, 2018.

\bibitem{georgiadis2006resource}
Leonidas Georgiadis, Michael~J Neely, Leandros Tassiulas, et~al.
\newblock Resource allocation and cross-layer control in wireless networks.
\newblock {\em Foundations and Trends in Networking}, 1(1):1--144, 2006.

\bibitem{karp1972reducibility}
Richard~M Karp.
\newblock Reducibility among combinatorial problems.
\newblock In {\em Complexity of computer computations}, pages 85--103.
  Springer, 1972.

\bibitem{karush1939minima}
William Karush.
\newblock Minima of functions of several variables with inequalities as side
  conditions.
\newblock {\em Master's Thesis, Department of Mathematics, University of
  Chicago}, 1939.

\bibitem{kleinberg2011load}
Robert Kleinberg, Georgios Piliouras, and {\'E}va Tardos.
\newblock Load balancing without regret in the bulletin board model.
\newblock {\em Distributed Computing}, 24(1):21--29, 2011.

\bibitem{kuhn1951}
Harold~W Kuhn and Albert~W Tucker.
\newblock Nonlinear programming.
\newblock In {\em Proceedings of the Second Berkeley Symposium on Mathematical
  Statistics and Probability}, pages 481--492, Berkeley, Calif., 1951.
  University of California Press.
\newblock URL: \url{https://projecteuclid.org/euclid.bsmsp/1200500249}.

\bibitem{Lamport78}
Leslie Lamport.
\newblock Time, clocks, and the ordering of events in a distributed system.
\newblock {\em Commun. {ACM}}, 21(7):558--565, 1978.
\newblock \href {https://doi.org/10.1145/359545.359563}
  {\path{doi:10.1145/359545.359563}}.

\bibitem{Lehmann81}
Daniel Lehmann and Michael~O Rabin.
\newblock On the advantages of free choice: {A} symmetric and fully distributed
  solution to the dining philosophers problem.
\newblock In John White, Richard~J Lipton, and Patricia~C Goldberg, editors,
  {\em Conference Record of the 8th Annual {ACM} Symposium on Principles of
  Programming Languages}, pages 133--138. {ACM} Press, 1981.
\newblock \href {https://doi.org/10.1145/567532.567547}
  {\path{doi:10.1145/567532.567547}}.

\bibitem{lenzen2011tight}
Christoph Lenzen and Roger Wattenhofer.
\newblock Tight bounds for parallel randomized load balancing.
\newblock In {\em Proceedings of the 43rd annual ACM Symposium on Theory of
  Computing}, pages 11--20, 2011.

\bibitem{lu2011join}
Yi~Lu, Qiaomin Xie, Gabriel Kliot, Alan Geller, James~R Larus, and Albert
  Greenberg.
\newblock Join-idle-queue: A novel load balancing algorithm for dynamically
  scalable web services.
\newblock {\em Performance Evaluation}, 68(11):1056--1071, 2011.

\bibitem{luczak2006maximum}
Malwina~J Luczak, Colin McDiarmid, et~al.
\newblock On the maximum queue length in the supermarket model.
\newblock {\em The Annals of Probability}, 34(2):493--527, 2006.

\bibitem{youtube_lec}
Tyler McMullen.
\newblock {Load Balancing is Impossible. Scaleconf 2016.}
\newblock \url{https://www.youtube.com/watch?v=kpvbOzHUakA}.

\bibitem{mitzenmacher2000useful}
Michael Mitzenmacher.
\newblock How useful is old information?
\newblock {\em IEEE Transactions on Parallel and Distributed Systems},
  11(1):6--20, 2000.

\bibitem{mitzenmacher2001power}
Michael Mitzenmacher.
\newblock The power of two choices in randomized load balancing.
\newblock {\em IEEE Transactions on Parallel and Distributed Systems},
  12(10):1094--1104, 2001.

\bibitem{mitzenmacher2016analyzing}
Michael Mitzenmacher.
\newblock Analyzing distributed join-idle-queue: A fluid limit approach.
\newblock In {\em 2016 54th Annual Allerton Conference on Communication,
  Control, and Computing}, pages 312--318. IEEE, 2016.

\bibitem{mitzenmacher2002load}
Michael Mitzenmacher, Balaji Prabhakar, and Devavrat Shah.
\newblock Load balancing with memory.
\newblock In {\em 43rd Annual IEEE Symposium on Foundations of Computer
  Science.}, pages 799--808. IEEE, 2002.

\bibitem{moon2020acceltcp}
YoungGyoun Moon, SeungEon Lee, Muhammad~Asim Jamshed, and KyoungSoo Park.
\newblock {A}ccel{TCP}: {A}ccelerating {N}etwork {A}pplications with {S}tateful
  {TCP} {O}ffloading.
\newblock In {\em 17th USENIX Symposium on Networked Systems Design and
  Implementation (NSDI)}, pages 77--92, 2020.

\bibitem{murray2012large}
David Murray, Terry Koziniec, Kevin Lee, and Michael Dixon.
\newblock {L}arge {MTU}s and internet performance.
\newblock In {\em 13th IEEE International Conference on High Performance
  Switching and Routing}, pages 82--87, 2012.

\bibitem{neely2007optimal}
Michael~J Neely.
\newblock Optimal energy and delay tradeoffs for multiuser wireless downlinks.
\newblock {\em IEEE Transactions on Information Theory}, 53(9):3095--3113,
  2007.

\bibitem{neely2010stability}
Michael~J Neely.
\newblock Stability and capacity regions or discrete time queueing networks.
\newblock {\em arXiv preprint arXiv:1003.3396}, 2010.

\bibitem{nishtala2017hipster}
Rajiv Nishtala, Paul Carpenter, Vinicius Petrucci, and Xavier Martorell.
\newblock Hipster: Hybrid task manager for latency-critical cloud workloads.
\newblock In {\em IEEE International Symposium on High Performance Computer
  Architecture (HPCA)}, pages 409--420, 2017.

\bibitem{prekas2017zygos}
George Prekas, Marios Kogias, and Edouard Bugnion.
\newblock Zygos: Achieving low tail latency for microsecond-scale networked
  tasks.
\newblock In {\em 26th Symposium on Operating Systems Principles (SOSP)}, pages
  325--341, 2017.

\bibitem{Rabin80}
Michael~O Rabin.
\newblock {$N$}-process synchronization by $4 \log_2$ {$N$}-valued shared
  variables.
\newblock In {\em 21st Annual Symposium on Foundations of Computer Science},
  pages 407--410. {IEEE} Computer Society, 1980.
\newblock \href {https://doi.org/10.1109/SFCS.1980.26}
  {\path{doi:10.1109/SFCS.1980.26}}.

\bibitem{schurman2009user}
Eric Schurman and Jake Brutlag.
\newblock The user and business impact of server delays, additional bytes, and
  http chunking in web search.
\newblock In {\em Velocity Web Performance and Operations Conference}.
  O'Reilly, 2009.

\bibitem{shah2002use}
Devavrat Shah and Balaji Prabhakar.
\newblock The use of memory in randomized load balancing.
\newblock In {\em Proceedings IEEE International Symposium on Information
  Theory}, page 125, 2002.

\bibitem{netflixEdge}
Mike Smith.
\newblock {Netflix Technology Blog. Rethinking Netflix’s Edge Load Balancing.
  September 2018.}
\newblock
  \url{https://netflixtechblog.com/netflix-edge-load-balancing-695308b5548c}.

\bibitem{stolyar2015pull}
Alexander~L Stolyar.
\newblock Pull-based load distribution in large-scale heterogeneous service
  systems.
\newblock {\em Queueing Systems}, 80(4):341--361, 2015.

\bibitem{stolyar2017pull}
Alexander~L Stolyar.
\newblock Pull-based load distribution among heterogeneous parallel servers:
  the case of multiple routers.
\newblock {\em Queueing Systems}, 85(1-2):31--65, 2017.

\bibitem{haproxy_po2}
Willy Tarreau.
\newblock {HAProxy. Test Driving “Power of Two Random Choices” Load
  Balancing.}
\newblock \url{https://www.haproxy.com/blog/power-of-two-load-balancing/},
  published on February 15, 2019.

\bibitem{van2019hyper}
Mark van~der Boor, Sem Borst, and Johan van Leeuwaarden.
\newblock Hyper-scalable jsq with sparse feedback.
\newblock {\em Proceedings of the ACM on Measurement and Analysis of Computing
  Systems}, 3(1):1--37, 2019.

\bibitem{vargaftik2020lsq}
Shay Vargaftik, Isaac Keslassy, and Ariel Orda.
\newblock {LSQ}: {L}oad {B}alancing in {L}arge-{S}cale {H}eterogeneous
  {S}ystems {W}ith {M}ultiple {D}ispatchers.
\newblock {\em IEEE/ACM Transactions on Networking}, 2020.

\bibitem{vvedenskaya1996queueing}
Nikita~Dmitrievna Vvedenskaya, Roland~L'vovich Dobrushin, and
  Fridrikh~Izrailevich Karpelevich.
\newblock Queueing system with selection of the shortest of two queues: An
  asymptotic approach.
\newblock {\em Problemy Peredachi Informatsii}, 32(1):20--34, 1996.

\bibitem{wang2018distributed}
Chunpu Wang, Chen Feng, and Julian Cheng.
\newblock Distributed join-the-idle-queue for low latency cloud services.
\newblock {\em IEEE/ACM Transactions on Networking}, 26(5):2309--2319, 2018.

\bibitem{weber1978optimal}
Richard~R Weber.
\newblock On the optimal assignment of customers to parallel servers.
\newblock {\em Journal of Applied Probability}, 15(2):406--413, 1978.

\bibitem{winston1977optimality}
Wayne Winston.
\newblock Optimality of the shortest line discipline.
\newblock {\em Journal of Applied Probability}, 14(1):181--189, 1977.

\bibitem{ying2017power}
Lei Ying, Rayadurgam Srikant, and Xiaohan Kang.
\newblock The power of slightly more than one sample in randomized load
  balancing.
\newblock {\em Mathematics of Operations Research}, 42(3):692--722, 2017.

\bibitem{zhou2020asymptotically}
Xingyu Zhou, Ness Shroff, and Adam Wierman.
\newblock Asymptotically optimal load balancing in large-scale heterogeneous
  systems with multiple dispatchers.
\newblock {\em arXiv preprint arXiv:2002.08908}, 2020.

\bibitem{zhou2019heavy}
Xingyu Zhou, Jian Tan, and Ness Shroff.
\newblock Heavy-traffic delay optimality in pull-based load balancing systems:
  Necessary and sufficient conditions.
\newblock {\em ACM SIGMETRICS Performance Evaluation Review}, 47(1):5--6, 2019.

\end{thebibliography}


\appendix

\section{NP-Hardness of the Unsplittable Case}
\label{sec:NP-hard}
The unsplittable instance of the optimization problem is not computationally tractable.
It can be seen as a general variant of the known \textsc{Bin-Packing} problem. We prove by reduction from the classic \textsc{Partition} problem \cite{karp1972reducibility} that it entails an NP-hard problem. More precisely, we show the following.
\NPhard*

\begin{proof}
    We reduce from \textsc{Partition}, which is known to be NP-complete.
    Recall that in the \textsc{Partition} problem we are given a set $\setL$ of $M$~natural numbers $l_1,\dots,l_M\in\naturals$, and we need to decide whether a partition of~$\setL$ into two subsets with equal total weights exists.
    Formally, $\setL\in\textsc{Partition}$ if and only if $\exists \setL_1\subseteq \setL$ s.t. $\sum\limits_{l\in \setL_1}l = \sum\limits_{l\in \setL\setminus \setL_1}l$.
    
    Given an instance~$\setL$ to \textsc{Partition}, we create the following system consisting of $N=2$ servers and $M=|\setL|$ dispatchers.
    The system starts at a state in which the length of both queues is 0, \ie $\{Q_1,Q_2\}=\{0,0\}$, and for each dispatcher~$m\in \setM$, its arrivals are set to $l_m$ (\ie $a^{(m)}=l_m$).
    Obviously, the reduction is polynomial (in fact, it is linear) in~$|\setL|$. We are left to show that
    \[
        \mathbb{E}\lVert Q^*-\bar{Q}\rVert_2^2=0 \iff \setL\in\textsc{Partition}.
    \]
    Notice that $a=\sum\limits_{m\in\setM}a^{(m)} = \sum\limits_{l\in \setL}l$, hence, $Q^*=\langle\frac{a}{2},\frac{a}{2}\rangle$.
    
    \T{Direction 1.} Assume that $\mathbb{E}\lVert Q^*-\bar{Q}\rVert_2^2=0$. Since $\lVert Q^*-\bar{Q}\rVert_2^2 \ge 0$, it is immediate that $\mathbb{P}(Q^*=\bar{Q}) > 0$. In turn, this means that it is possible to equally divided the jobs between the two servers and therefore $\setL\in\textsc{Partition}$.
    
    \T{Direction 2.} Assume that $\setL\in\textsc{Partition}$. By definition, there exists $\setL_1\subseteq \setL$ such that $\sum\limits_{l\in \setL_1}l = \sum\limits_{l\in \setL\setminus \setL_1}l$. 
    For each~$l_m\in\setL_1$ dispatcher~$m$ sends its batch of jobs to server~$1$ with probability~$1$.
    Similarly, for each~$l_m\in\setL\setminus\setL_1$ dispatcher~$m$ sends its batch of jobs to server~$2$ with probability~$1$.
    This allocation ensures that $\mathbb{P}(Q^*=\bar{Q}) = 1$ and therefore $\lVert Q^*-\bar{Q}\rVert_2^2 \ge 0$.
\end{proof}


\section{Strong stability}\label{app:stabiltiy}

For convenience, we restate the relevant assumptions of our model.
We consider a discrete-time system with a set $\setM = \set{1,\ldots,M}$ of dispatchers distributing arriving jobs among a set $\setN = \set{1,\ldots,N}$ of servers. 
We denote by $a^{(m)}(t)$ the number of exogenous job arrivals at dispatcher $m$ at the beginning of round $t$ and denote $a(t) = \sum_{m \in \setM} a^{(m)}(t)$. We assume, for all $m$,
\begin{equation}\label{eq:arrivals}
    \brac{a^{(m)}(t)}_{t=0}^{\infty} \text{ is an $i.i.d.$ process}, \quad \E[a^{(m)}(0)] = \lambda^{(1)}, \quad \E [(a^{(m)}(0))^2] = \lambda^{(2)}.
\end{equation}
Each server has a FIFO queue for storing pending jobs. Let $Q_n(t)$ be the queue length of server~$n$ at the beginning of round $t$ (before any job arrivals and departures) and denote $Q(t) \triangleq  (Q_1(t),\dots,Q_N(t))$.
Let $s_n(t)$ be the potential service offered to the queue at server $n$ at round $t$. That is, $s_n(t)$ is the maximum number of jobs that can be completed by server $n$ at round $t$. We assume that, for all $n$,  
\begin{equation}\label{eq:departures}
    \{s_n(t)\}_{t=0}^{\infty} \text{ is an $i.i.d.$ process}, \quad \E[s_n(0)] = \mu^{(1)}, \quad \E [(s_n(0))^2] = \mu^{(2)}.
\end{equation}
We assume the total expected arrival rate to the system is admissible. That is, we assume that there exists an $\epsilon > 0$ such that $M \lambda^{(1)} = N \mu^{(1)} - \epsilon$.

We prove that our dispatching policies are strongly stable. Specifically, the strong stability proof we conduct applies to all our policies, \ie \TWLs, L-\TWLs and \TWLs$^{ts}$. Our proof follows the same lines as in \cite{vargaftik2020lsq} with a few key modifications that capture the somewhat different nature of our dispatching policies in which the dispatching probability to a specific server is \emph{dependent} on the arrival rate. 

\begin{definition}[Strong stability] A load balancing system is said to be strongly stable if for any admissible arrival rate it holds that 
$$\limsup_{T \to \infty} \frac{1}{T} \sum_{t=0}^{T-1} \sum_{n\in\setN} \E \Big[Q_n(t)\Big] < \infty~.$$
\end{definition}

Strong stability is a strong form of stability often employed in discrete-time queuing systems. Moreover, in our model, where we assume the existence of the first two moments of the arrival and departure processes, strong stability also implies throughput optimality as well as other commonly considered forms of stability, such as steady state stability, rate stability, mean rate stability and more (the reader is referred to \cite{neely2010stability,neely2007optimal,georgiadis2006resource} for details).

Let $\receives_n(t) = \sum_{m\in\setM} \receives_n^{(m)}(t)$. Then, the queue dynamics of server $n$ is given by
\begin{equation}\label{eq:queue_dynamics}
Q_n(t+1) = \max \{0, Q_n(t) + \receives_n(t) - s_n(t)\}~. 
\end{equation}
Squaring both sides of \eqref{eq:queue_dynamics}, rearranging, and omitting terms yields, 
\begin{equation}\label{eq:basic_nneq}
    \bp{Q_n(t+1)}^2 - \bp{Q_n(t)}^2 \le \bp{\receives_n(t)}^2 + \bp{s_n(t)}^2 - 2Q_n(t)\bp{s_n(t)-\receives_n(t)}~.    
\end{equation}
Summing \eqref{eq:basic_nneq} over the servers yields
\begin{equation}\label{eq:sum_over_servers}
  \sum_{n\in\setN} \bp{Q_n(t+1)}^2 - \sum_{n\in\setN} \bp{Q_n(t)}^2 \le \sum_{n\in\setN} \bp{\receives_n(t)}^2 + \sum_{n\in\setN} \bp{s_n(t)}^2 - 2\sum_{n\in\setN} Q_n(t)\bp{s_n(t)-\receives_n(t)}~.
\end{equation}

We now split the proof into the unsplittable and then the splittable case.

\subsection{The unsplittable case}

For each $(n,m) \in \setN \times \setM$, let $I_n^{(m)}(t)$ be an indicator function that takes the value of 1 with probability $1/N$ and 0 otherwise such that $\sum_{n\in\setN} I_n^{(m)}(t) = 1 \quad \forall m \in \setM$. We now rewrite \eqref{eq:sum_over_servers} by using $\receives_n(t) = \sum_{m\in\setM} \receives_n^{(m)}(t)$ and then adding and subtracting the term $2 \sum_{n\in\setN}  \sum_{m\in\setM} I_n^{(m)}(t) a^{(m)}(t) Q_n(t)$ from the right hand side of the equation. This yields
\begin{equation}\label{eq:upsplittable_abc}
\begin{split}
  & \sum_{n\in\setN} \bp{Q_n(t+1)}^2 - \sum_{n\in\setN} \bp{Q_n(t)}^2 \le \underbrace{\sum_{n\in\setN} \bp{\receives_n(t)}^2 + \sum_{n\in\setN} \bp{s_n(t)}^2}_{(a)} \cr
  &  - 2\underbrace{\sum_{n\in\setN} Q_n(t)\bp{s_n(t)-\sum_{m\in\setM} I_n^{(m)}(t) a^{(m)}(t)}}_{(b)} \cr
  &+ 2\underbrace{\sum_{n\in\setN} Q_n(t)\bp{\sum_{m\in\setM} \receives_n^{(m)}(t)-\sum_{m\in\setM} I_n^{(m)}(t) a^{(m)}(t)}}_{(c)}.
\end{split}
\end{equation}

We now analyze term (a) in \eqref{eq:upsplittable_abc}. Taking expectation and using \eqref{eq:arrivals} and \eqref{eq:departures} we obtain

\begin{equation}\label{eq:unsplittable_a}
\begin{split}
&\mathbb{E} \bigg[\sum_{n\in\setN} \bp{\receives_n(t)}^2 + \sum_{n\in\setN} \bp{s_n(t)}^2\bigg] \le \mathbb{E} \bigg[\bp{\sum_{n\in\setN} \receives_n(t)}^2 \bigg] + N\MuTwo \cr
& =  \mathbb{E} \bigg[\bp{\sum_{m\in\setM} a^{(m)}(t)}^2 \bigg] + N\MuTwo = M\LambdaTwo + M(M-1)(\LambdaOne)^2 + N\MuTwo \triangleq A.
\end{split}
\end{equation}

Next, for (b) in \eqref{eq:upsplittable_abc}, taking expectation and using \eqref{eq:arrivals}, \eqref{eq:departures}, the definition of $I_n^{(m)}(t)$ and the admissibility of the system, we obtain
\begin{equation}\label{eq:unsplittable_b}
\begin{split}
\mathbb{E} \sum_{n\in\setN} Q_n(t)\bp{s_n(t)-\sum_{m\in\setM} I_n^{(m)}(t) a^{(m)}(t)} =  \sum_{n\in\setN} \bp{\MuOne-\frac{M\LambdaOne}{N}} \mathbb{E} [Q_n(t)] = \frac{\epsilon}{N} \sum_{n\in\setN} \mathbb{E} [Q_n(t)].
\end{split}
\end{equation}

Let $\tilde{Q}_n^{(m)}(t)$ be the local state of server $n$ at dispatcher $m$ at the beginning of round $t$. For complete information, it trivially holds that $\tilde{Q}_n^{(m)}(t)=Q_n^{(m)}(t)$. Otherwise, there exists a constant $C$, such that it holds $\mathbb{E} |Q_n(t) - \tilde{Q}_n^{(m)}(t)| \le C$ for all $n,m,t$. Indeed, this holds for all the local state (\ie array) updates for any $\rho>0$.\footnote{Recall that $\rho$ is the fraction of server each dispatcher samples uniformly at random during the communication phase (\ie phase 4) of each round.} (full derivation can be found in the proof of Theorem 2, followed by Proposition 1 in \cite{vargaftik2020lsq}).

We next turn to analyze term (c) in \eqref{eq:upsplittable_abc}. We substitute $Q_n(t)$ by $Q_n(t)-\tilde{Q}_n^{(m)}(t)+\tilde{Q}_n^{(m)}(t)$ and rearrange. This yields
\begin{equation}\label{eq:add_substruct_trick}
\begin{split}
  &\sum_{n\in\setN} Q_n(t)\bp{\sum_{m\in\setM} \receives_n^{(m)}(t)-\sum_{m\in\setM} I_n^{(m)}(t) a^{(m)}(t)} = \cr
  & \sum_{n\in\setN} \sum_{m\in\setM} \tilde{Q}_n^{(m)}(t) (\receives_n^{(m)}(t)-I_n^{(m)}(t) a^{(m)}(t)) + \cr
  & \sum_{n\in\setN} \sum_{m\in\setM} (Q_n(t) - \tilde{Q}_n^{(m)}(t)) (\receives_n^{(m)}(t)-I_n^{(m)}(t) a^{(m)}(t)).
\end{split}
\end{equation}
Now, we change the order of summation, use the triangle inequality and take the expectation of \eqref{eq:add_substruct_trick}. This yields,
\begin{equation}\label{eq:add_substruct_trick_expectation}
\begin{split}
  &\mathbb{E} \sum_{n\in\setN} Q_n(t)\bp{\sum_{m\in\setM} \receives_n^{(m)}(t)-\sum_{m\in\setM} I_n^{(m)}(t) a^{(m)}(t)} \le \cr
  & \sum_{m\in\setM} \mathbb{E} \sum_{n\in\setN} \tilde{Q}_n^{(m)}(t) (\receives_n^{(m)}(t)-I_n^{(m)}(t) a^{(m)}(t)) + \cr
  &\sum_{n\in\setN} \sum_{m\in\setM} \mathbb{E} |Q_n(t) - \tilde{Q}_n^{(m)}(t)| |\receives_n^{(m)}(t)-I_n^{(m)}(t) a^{(m)}(t)|.
\end{split}
\end{equation}

We now handle the two summation terms separately. By the linearity of expectation and the law of total expectation we have, 

\begin{equation}
\begin{split}
&\mathbb{E} \sum_{n\in\setN} \tilde{Q}_n^{(m)}(t) \receives_n^{(m)}(t) - \mathbb{E} \sum_{n\in\setN} \tilde{Q}_n^{(m)}(t) I_n^{(m)}(t) a^{(m)}(t) = \cr 
&\mathbb{E} \sum_{n\in\setN} \tilde{Q}_n^{(m)}(t) \mathbb{E} [\receives_n^{(m)}(t) \big| \tilde{Q}^{(m)}(t), a^{(m)}(t)] - \mathbb{E} \sum_{n\in\setN} \tilde{Q}_n^{(m)}(t) \frac{a^{(m)}(t)}{N}.
\end{split}
\end{equation}

Now, by the definition of our policy it holds that

(i)~$\tilde{Q}_n^{(m)}(t) \le \tilde{Q}_{n'}^{(m)}(t)$ iff $\mathbb{E} [\receives_n^{(m)}(t) \big| \tilde{Q}^{(m)}(t), a^{(m)}(t)] \ge \mathbb{E} [\receives_{n'}^{(m)}(t) \big| \tilde{Q}^{(m)}(t), a^{(m)}(t)]$ for any $(n,{n'}) \in \setN \times \setN$. This is because by \eqref{eq:unsplittableWfProbabilities}, $$\mathbb{E} [\receives_n^{(m)}(t) \big| \tilde{Q}^{(m)}(t), a^{(m)}(t)] = a^{(m)}(t) \cdot \max \left\{ 0,\frac{\receives_n^* - (a^{(m)}-\sum_{k \notin U} \receives_k^*) / {|\setU|} } {(M-1)a^{(m)}(t)} \right\}.$$ 
This term is monotonically increasing in $\receives_n^* =\textsc{WaterLevel}(\tilde{Q}^{(m)}(t),M a^{(m)}(t)) - \tilde{Q}_n^{(m)}(t)$. Thus, it is monotonically decreasing in $\tilde{Q}_n^{(m)}(t)$ and the claim holds. It also holds that 

(ii)~$\sum_{n\in\setN} \mathbb{E} [\receives_n^{(m)}(t) \big| \tilde{Q}^{(m)}(t), a^{(m)}(t)] = a^{(m)}(t) = \sum_{n\in\setN} \frac{a^{(m)}(t)}{N}$. Recall that each dispatcher $m$ sends jobs according to \uTWF{} based on its local array $\tilde{Q}^{(m)}(t)$. Now, observe that the vector $\{ \tilde{Q}_n^{(m)}(t) \mathbb{E} [\receives_n^{(m)}(t) \big| \tilde{Q}^{(m)}(t), a^{(m)}(t)] \}_{n\in\setN}$ is \emph{majorized} by the vector $\{ \tilde{Q}_n^{(m)}(t) \frac{a^{(m)}(t)}{N} \}_{n\in \setN}$.%
\footnote{A similar majorization argument (which we also use in the splittable case) was recently applied by \cite{zhou2020asymptotically} to prove stability in a similar setting with local states.
However, their criteria of \emph{tilted} dispatching policies does not capture our dependence of the dispatching on the arrival process and thus cannot be directly applied in our model.} 
Therefore, 
\begin{equation}\label{eq:left_add_substruct_trick_expectation}
\mathbb{E} \sum_{n\in\setN} \tilde{Q}_n^{(m)}(t) (\receives_n^{(m)}(t)-I_n^{(m)}(t) a^{(m)}(t)) \le 0 \quad \forall m \in \setM. 
\end{equation}

We next handle the second term of \eqref{eq:add_substruct_trick_expectation}. Recall that all our local array update policies respect that $\mathbb{E} |Q_n(t) - \tilde{Q}_n^{(m)}(t)| \le C$ for any $t$ independently of the arrivals, we obtain,
\begin{equation}\label{eq:right_add_substruct_trick_expectation}
\sum_{n\in\setN} \sum_{m\in\setM} \mathbb{E} |Q_n(t) - \tilde{Q}_n^{(m)}(t)| |\receives_n^{(m)}(t)-I_n^{(m)}(t) a^{(m)}(t)| \le \LambdaOne NMC.
\end{equation}

Using \eqref{eq:left_add_substruct_trick_expectation} and \eqref{eq:right_add_substruct_trick_expectation} in \eqref{eq:add_substruct_trick_expectation} yields,

\begin{equation}\label{eq:unsplittable_c}
 \mathbb{E} \sum_{n\in\setN} Q_n(t)\bp{\sum_{m\in\setM} \receives_n^{(m)}(t)-\sum_{m\in\setM} I_n^{(m)}(t) a^{(m)}(t)} \le  0 + \LambdaOne NMC = \LambdaOne NMC. 
\end{equation}
Finally, taking the expectation of \eqref{eq:upsplittable_abc} as well as using \eqref{eq:unsplittable_a}, \eqref{eq:unsplittable_b}, \eqref{eq:unsplittable_c} and rearranging yields, 
\begin{equation}\label{eq:upsplittable_abc_expectation}
\begin{split}
  & \mathbb{E} \bigg[ \sum_{n\in\setN} \bp{Q_n(t+1)}^2 - \sum_{n\in\setN} \bp{Q_n(t)}^2 \bigg] \le A + 2\LambdaOne NMC - 2\frac{\epsilon}{N} \sum_{n\in\setN} \mathbb{E} [Q_n(t)].
\end{split}
\end{equation}

Next, summing \eqref{eq:upsplittable_abc_expectation} over rounds $0,\ldots,T{-}1$, multiplying by $\frac{N}{2\epsilon T}$ and rearranging yields,
\begin{equation}\label{eq:upsplittable_abc_expectation_simple}
\frac{1}{T} \sum_{t=0}^{T-1} \sum_{n\in\setN} \E \Big[Q_n(t)\Big] \le \frac{AN + 2\LambdaOne N^2 MC}{2\epsilon} + \frac{N}{2\epsilon T} \E \sum_{n\in\setN} \bp{Q_n(0)}^2,
\end{equation}
where we omitted the non-positive term $\E \Big[-\sum_{n\in\setN} \bp{Q_n(T)}^2\Big]$ as a results of the telescopic series at the left hand side of \eqref{eq:upsplittable_abc_expectation}. Taking limits of \eqref{eq:upsplittable_abc_expectation_simple} and making the standard assumption that the system is initialized with bounded queue lengths, \ie \mbox{$\E \Big[\sum_{n\in\setN} \bp{Q_n(0)}^2\Big] < \infty$} yields, 
\begin{equation}\label{eq:unsplittable_final}
\begin{split}
    &  \limsup_{T \to \infty} \frac{1}{T} \sum_{t=0}^{T-1} \sum_{n\in\setN} \E \Big[Q_n(t)\Big] \le \frac{AN + 2\LambdaOne N^2MC}{2\epsilon}.
\end{split}
\end{equation}
This concludes the proof for the unsplittable case.

\subsection{The splittable case}

In the splittable case, each job can be sent to a different server. For completeness, we rewrite the proof of the unsplittable case to capture this property. For each $(n,m,k)$, let $I_n^{m,k}(t)$ be an indicator function that takes the value of 1 with probability $1/N$ and 0 otherwise such that $\sum_{n\in\setN} I_n^{m,k}(t) = 1 \quad \forall m \in \setM, k \in [1,\ldots,a^{(m)}(t)]$. We now again rewrite \eqref{eq:sum_over_servers} but this time adding and subtracting the term $2 \sum_{n\in\setN}  \sum_{m\in\setM} \sum_{k=1}^{a^{(m)}(t)} I_n^{m,k}(t) Q_n(t)$ from the right hand side of the equation. This yields, 
\begin{equation}\label{eq:splittable_abc}
\begin{split}
  & \sum_{n\in\setN} \bp{Q_n(t+1)}^2 - \sum_{n\in\setN} \bp{Q_n(t)}^2 \le \underbrace{\sum_{n\in\setN} \bp{\receives_n(t)}^2 + \sum_{n\in\setN} \bp{s_n(t)}^2}_{(a)} \cr
  &  - 2\underbrace{\sum_{n\in\setN} Q_n(t)\bp{s_n(t)-\sum_{m\in\setM} \sum_{k=1}^{a^{(m)}(t)} I_n^{m,k}(t)}}_{(b)} \cr
  & + 2\underbrace{\sum_{n\in\setN} Q_n(t)\bp{\sum_{m\in\setM} \receives_n^{(m)}(t)-\sum_{m\in\setM} \sum_{k=1}^{a^{(m)}(t)} I_n^{m,k}(t)}}_{(c)}.
\end{split}
\end{equation}

Term (a) in \eqref{eq:splittable_abc} is exactly as in the unsplittable case and its analysis is unchanged. For term (b), while $I_n^{(m)}(t) a^{(m)}(t)$ is replaced by $\sum_{k=1}^{a^{(m)}(t)} I_n^{m,k}(t)$, taking expectation and using \eqref{eq:arrivals}, \eqref{eq:departures}, the definition of $I_n^{m,k}(t)$ and the admissibility of the system, still yields the same result. that is,
\begin{equation}\label{eq:splittable_b}
\begin{split}
\mathbb{E} \sum_{n\in\setN} Q_n(t)\bp{s_n(t)-\sum_{m\in\setM} \sum_{k=1}^{a^{(m)}(t)} I_n^{m,k}(t)} =  \sum_{n\in\setN} \bp{\MuOne-\frac{M\LambdaOne}{N}} \mathbb{E} [Q_n(t)] = \frac{\epsilon}{N} \sum_{n\in\setN} \mathbb{E} [Q_n(t)].
\end{split}
\end{equation}
This is because $a^{(m)}(t)$ and $I_n^{m,k}(t)$ are independent which allows us to use Wald's identity. That is, we use $\mathbb{E} [\sum_{k=1}^{a^{(m)}(t)} I_n^{m,k}(t)] = \mathbb{E} [a^{(m)}(t)] \mathbb{E} [I_n^{m,1}(t)]$. 

We next turn to analyze term (c) in \eqref{eq:splittable_abc} by repeating the same steps as in the unsplittable case but with substituting $I_n^{(m)}(t) a^{(m)}(t)$ with $\sum_{k=1}^{a^{(m)}(t)} I_n^{m,k}(t)$. We again substitute $Q_n(t)$ by $Q_n(t)-\tilde{Q}_n^{(m)}(t)+\tilde{Q}_n^{(m)}(t)$ and rearrange terms. This yields,
\begin{equation}\label{eq:add_substruct_trick_split}
\begin{split}
  &\sum_{n\in\setN} Q_n(t)\bp{\sum_{m\in\setM} \receives_n^{(m)}(t)-\sum_{m\in\setM} \sum_{k=1}^{a^{(m)}(t)} I_n^{m,k}(t)} = \cr
  & \sum_{n\in\setN} \sum_{m\in\setM} \tilde{Q}_n^{(m)}(t) (\receives_n^{(m)}(t)-\sum_{k=1}^{a^{(m)}(t)} I_n^{m,k}(t)) + \cr
  &\sum_{n\in\setN} \sum_{m\in\setM} (Q_n(t) - \tilde{Q}_n^{(m)}(t)) (\receives_n^{(m)}(t)-\sum_{k=1}^{a^{(m)}(t)} I_n^{m,k}(t)).
\end{split}
\end{equation}
Now, we change the order of summation, use the triangle inequality and take the expectation of \eqref{eq:add_substruct_trick_split}. This yields
\begin{equation}\label{eq:add_substruct_trick_expectation_split}
\begin{split}
  &\mathbb{E} \sum_{n\in\setN} Q_n(t)\bp{\sum_{m\in\setM} \receives_n^{(m)}(t)-\sum_{m\in\setM} \sum_{k=1}^{a^{(m)}(t)} I_n^{m,k}(t)} \le \cr
  & \sum_{m\in\setM} \mathbb{E} \sum_{n\in\setN} \tilde{Q}_n^{(m)}(t) (\receives_n^{(m)}(t)-\sum_{k=1}^{a^{(m)}(t)} I_n^{m,k}(t)) + \cr
  &\sum_{n\in\setN} \sum_{m\in\setM} \mathbb{E} |Q_n(t) - \tilde{Q}_n^{(m)}(t)| |\receives_n^{(m)}(t)-\sum_{k=1}^{a^{(m)}(t)} I_n^{m,k}(t)|.
\end{split}
\end{equation}

We begin by handling the first summation term of \eqref{eq:add_substruct_trick_expectation_split}. By the linearity of expectation, the law of total expectation and the independence of $\tilde{Q}^{(m)}(t)$ and $I_n^{m,k}(t)$ we have
\begin{equation}
\begin{split}
&\mathbb{E} \sum_{n\in\setN} \tilde{Q}_n^{(m)}(t) \receives_n^{(m)}(t) - \mathbb{E} \sum_{n\in\setN} \tilde{Q}_n^{(m)}(t) \sum_{k=1}^{a^{(m)}(t)} I_n^{m,k}(t)= \cr &\mathbb{E} \sum_{n\in\setN} \bigg(\tilde{Q}_n^{(m)}(t) \mathbb{E} [\receives_n^{(m)}(t) \big| \tilde{Q}^{(m)}(t), a^{(m)}(t)]\bigg) - \mathbb{E} \sum_{n\in\setN} \bigg( \tilde{Q}_n^{(m)}(t) \mathbb{E} [\sum_{k=1}^{a^{(m)}(t)}  I_n^{m,k}(t) \big| a^{(m)}(t)]\bigg).
\end{split}
\end{equation}

Now, by the definition of our splittable policy it also holds that 

(i)~$\tilde{Q}_n^{(m)}(t) \le \tilde{Q}_{n'}^{(m)}(t)$ iff $\mathbb{E} [\receives_n^{(m)}(t) \big| \tilde{Q}^{(m)}(t), a^{(m)}(t)] \ge \mathbb{E} [\receives_{n'}^{(m)}(t)\big| \tilde{Q}^{(m)}(t), a^{(m)}(t)]$ for any $(n,n') \in \setN \times \setN$. This is because $$\mathbb{E} [\receives_n^{(m)}(t) \big| \tilde{Q}^{(m)}(t), a^{(m)}(t)] = \sum_{k=1}^{a^{(m)}(t)} max \{0, \frac{\receives_n^*-1/\underWL}{a^{(m)}(t) M -1}\}.$$
This term is monotonically increasing in $\receives_n^* =\textsc{WaterLevel}(\tilde{Q}^{(m)}(t),M a^{(m)}(t)) - \tilde{Q}_n^{(m)}(t)$. Thus, it is monotonically decreasing in $\tilde{Q}_n^{(m)}(t)$ and the claim holds. It also holds that, 

(ii)~$\sum_{n\in\setN} \mathbb{E} [\receives_n^{(m)}(t) \big| \tilde{Q}^{(m)}(t), a^{(m)}(t)] = a^{(m)}(t) = \sum_{n\in\setN} \mathbb{E} [\sum_{k=1}^{a^{(m)}(t)} I_n^{m,k}(t) \big| a^{(m)}(t)]$. 

Therefore, the vector $\{ \tilde{Q}_n^{(m)}(t) \mathbb{E} [\receives_n^{(m)}(t) \big| \tilde{Q}^{(m)}(t), a^{(m)}(t)] \}_{n\in\setN}$ is \emph{majorized} by the vector $\tilde{Q}_n^{(m)}(t) \mathbb{E} [\sum_{k=1}^{a^{(m)}(t)}  I_n^{m,k}(t) \big| a^{(m)}(t)] \}_{n\in\setN}$. This is because $\mathbb{E} [\sum_{k=1}^{a^{(m)}(t)}  I_n^{m,k}(t) \big| a^{(m)}(t)] = \frac{a^{(m)}(t)}{N}$ for all $n \in \setN$. Therefore, 
\begin{equation}\label{eq:left_add_substruct_trick_expectation_split}
\mathbb{E} \sum_{n\in\setN} \tilde{Q}_n^{(m)}(t) (\receives_n^{(m)}(t)-\sum_{k=1}^{a^{(m)}(t)} I_n^{m,k}(t)) \le 0 \quad \forall m \in \setM. 
\end{equation}

Now, we handle the second summation term of \eqref{eq:add_substruct_trick_expectation_split}. Recall once more that $\mathbb{E} |Q_n(t) - \tilde{Q}_n^{(m)}(t)| \le C$ for any $t$ and independently of the arrivals. Hence, we obtain,
\begin{equation}\label{eq:right_add_substruct_trick_expectation_split}
\sum_{n\in\setN} \sum_{m\in\setM} \mathbb{E} |Q_n(t) - \tilde{Q}_n^{(m)}(t)| |\receives_n^{(m)}(t)-\sum_{k=1}^{a^{(m)}(t)} I_n^{m,k}(t)| \le \LambdaOne NMC.
\end{equation}

Using \eqref{eq:left_add_substruct_trick_expectation_split} and \eqref{eq:right_add_substruct_trick_expectation_split} in \eqref{eq:add_substruct_trick_expectation_split} yields,

\begin{equation}\label{eq:splittable_c}
 \mathbb{E} \sum_{n\in\setN} Q_n(t)\bp{\sum_{m\in\setM} \receives_n^{(m)}(t)-\sum_{m\in\setM} \sum_{k=1}^{a^{(m)}(t)} I_n^{m,k}(t)} \le \LambdaOne NMC. 
\end{equation}
Finally, taking the expectation of \eqref{eq:splittable_abc} as well as using \eqref{eq:unsplittable_a}, \eqref{eq:splittable_b}, \eqref{eq:splittable_c} and rearranging yields 
\begin{equation}\label{eq:upsplittable_abc_expectation_split}
\begin{split}
  & \mathbb{E} \bigg[ \sum_{n\in\setN} \bp{Q_n(t+1)}^2 - \sum_{n\in\setN} \bp{Q_n(t)}^2 \bigg] \le A + 2\LambdaOne NMC - 2\frac{\epsilon}{N} \sum_{n\in\setN} \mathbb{E} [Q_n(t)].
\end{split}
\end{equation}

The remaining analysis is the same as in the unsplittable case.
This concludes the proof.


\section{Additional Evaluation}\label{app:moreevaluation}

This section provides additional evaluation results, including run-time measurements, heavy-tailed arrivals, and larger systems.

\subsection{Runtime measurements}\label{app:rtmeasuments}

In TWF, a dispatcher computes the water level according to Algorithm \ref{alg:WaterLevel}, which runs in $O(\min(N,a))$ time complexity if the queues are presorted. Consequently, the time complexity of computing the dispatching assignment is dominated by sorting, i.e., $O(N\log N)$. Asymptotically, this is equivalent to the complexity of commonly used algorithms such as $JSQ$, which rely on sorting. 
To illustrate that TWF admits computational overheads similar to that of $JSQ$, we have conducted runtime measurements in a challenging and more time-consuming splittable scenario with a load of $\rho=0.99$, ten dispatchers, and an increasing number of servers.

The results depicted in Figure \ref{fig:runtime} show that TWF scales similarly to $JSQ$ as the number of servers increases (note that the runtime complexity depends only on the number of servers).
This is evident to the practicality of TWF since $JSQ$ is not reported to cause running time issues in large-scale systems.

\begin{figure*}[h]
\centering
\includegraphics[width=0.99\textwidth]{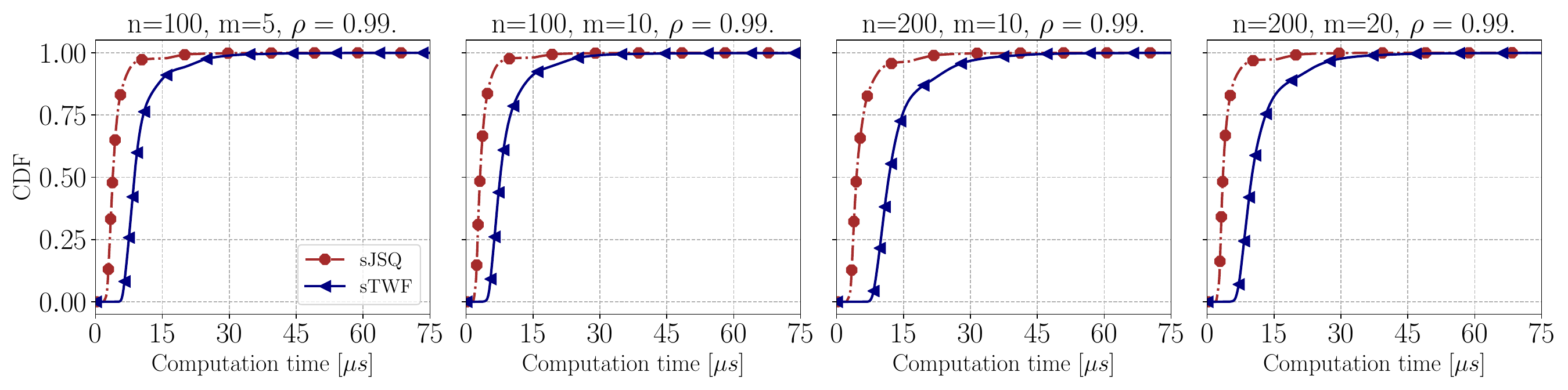}
\caption{Runtime measurements of JSQ and TWF.}
\label{fig:runtime}
\end{figure*}

\subsection{Heavy-tailed arrivals}\label{app:htarrivalseval}

We next explore the effect of a heavy tail (Log-normal) distribution. Specifically, we focus on the splittable case and set the Log-normal distribution parameters to~$\mu=0$ and~$\sigma$ to satisfy the desired load~$\rho$. 
Figure~\ref{fig:Lognormal mean} shows the resulting average response times for four different systems.
Figure~\ref{fig:Lognormal tail} shows the resulting tail distributions for the same different systems. Here, the gap between TWF and JSQ somewhat decreases. However, the performance is still significantly in favor of TWF.

\begin{figure}
     \centering
     \includegraphics[width=\textwidth]{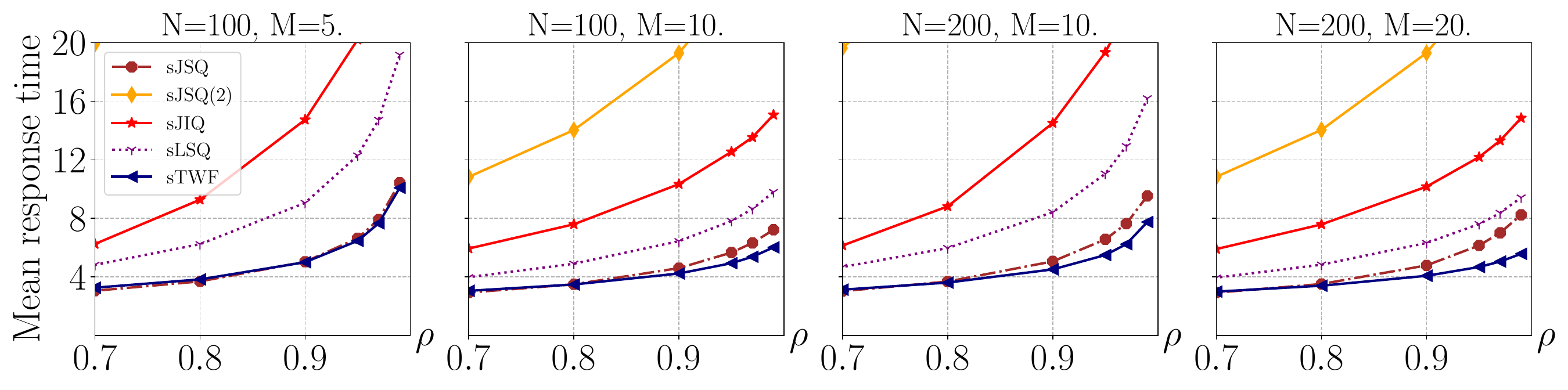}
     \caption{Heavy tail (Log-normal) arrival processes. Average job response times as a function of the load over four different systems. The $x$-axis represents load~$\rho$. The $y$-axis represents the average response time.}
     \label{fig:Lognormal mean}
\end{figure}

\begin{figure}
     \centering
     \includegraphics[width=\textwidth]{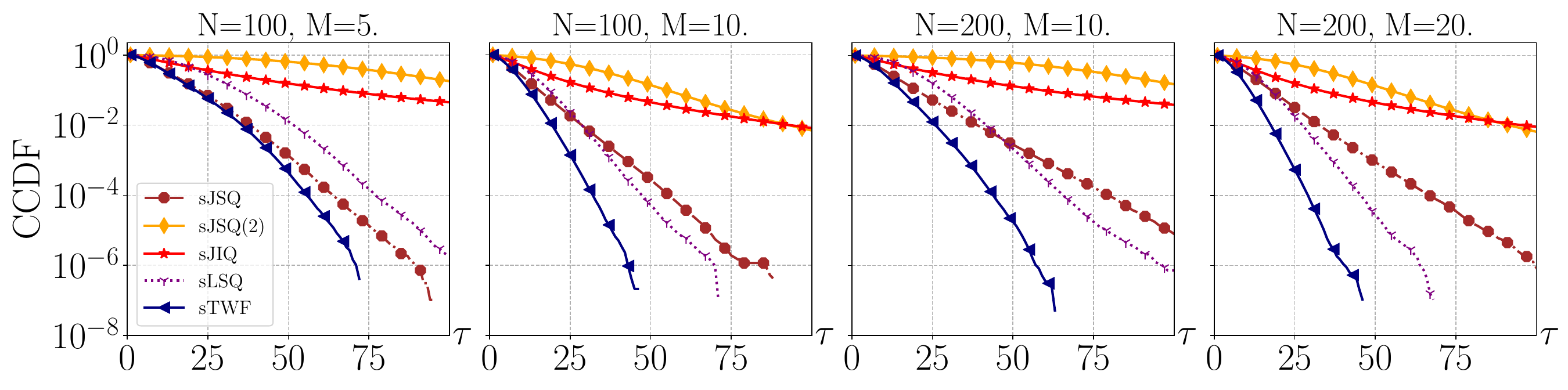}
     \caption{Heavy tail (Log-normal) arrival processes. Response-time tail distributions at high load ($\rho=0.99$). The $x$-axis represents the response time (denoted by $\tau$). The $y$-axis represents the CCDF.}
     \label{fig:Lognormal tail}
\end{figure}

\subsection{Larger systems}\label{app:largersystemseval}

We next conduct simulations with 500 servers. The results are shown in Figures~\ref{fig:n500 mean} and~\ref{fig:n500 tail} and decisively show that TWF significantly improves performance in larger-scale systems.

\begin{figure}
     \centering
     \includegraphics[width=\textwidth] {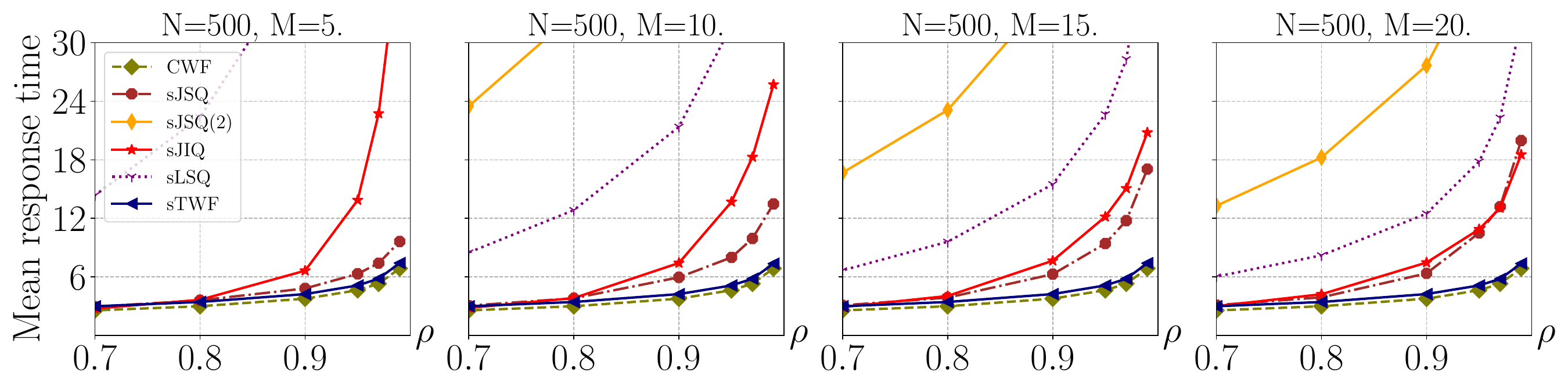}
     \caption{Comparison over a larger system. Centralized water-filling (CWF) is used as a benchmark to exhibit the loss due to the distribution. The figure shows the average job response time as a function of the load over four different systems. The $x$-axis represents load $\rho$. The $y$-axis represents the average response time.}
     \label{fig:n500 tail}
\end{figure}

\begin{figure}
     \centering
     \includegraphics[width=\textwidth]{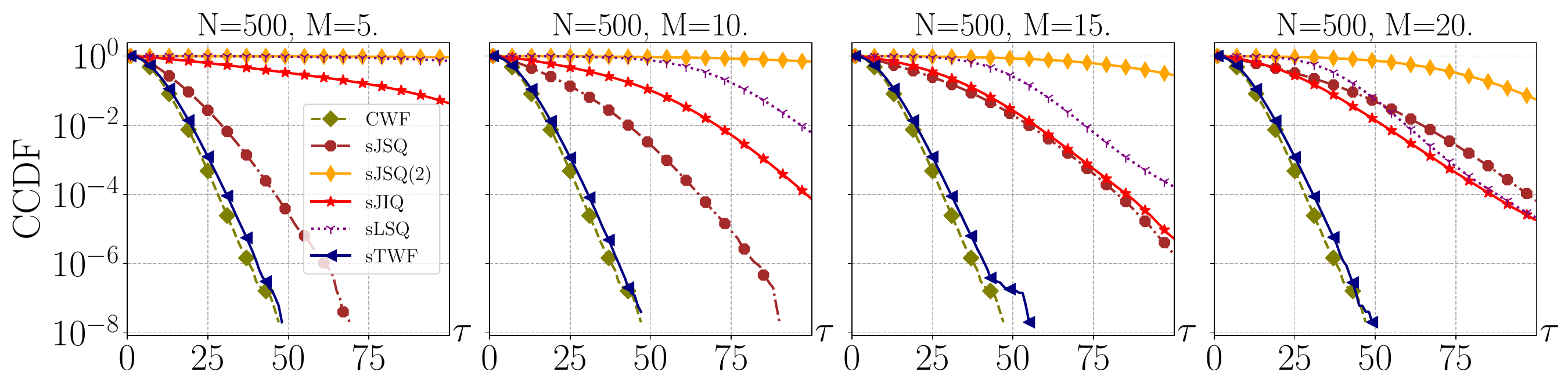}
     \caption{Comparison over a larger system. Centralized water-filling (CWF) is used as a benchmark to exhibit the loss due to the distribution. Response-time tail distributions at high load ($\rho=0.99$). The $x$-axis represents the response time (denoted by $\tau$). The $y$-axis represents the CCDF.}
     \label{fig:n500 mean}
\end{figure}

\end{document}